\documentclass[a4paper,twoside,12pt]{article}

\usepackage{amsmath}
\usepackage{amssymb}
\usepackage{amsfonts}
\usepackage{amsthm}
\usepackage{bm}
\usepackage{dsfont}
\usepackage{physics}
\usepackage{algorithm2e,setspace}
\RestyleAlgo{boxruled}
\usepackage{cprotect}

\usepackage{tikz}
\usetikzlibrary{cd}
\tikzcdset{scale cd/.style={every label/.append style={scale=#1},
    cells={nodes={scale=#1}}}}

\usetikzlibrary{hobby}
\usetikzlibrary{arrows,decorations.markings,arrows.meta,patterns}
\usetikzlibrary{perspective,3d}
\tikzset
{
   ->-/.style={decoration={markings,mark=at position 0.5 with {\arrow{Straight Barb}}},
               postaction={decorate}}
}
\usepackage{pgffor}
\usetikzlibrary{calc}
\usepackage{tikz-3dplot}

\usepackage[backend=biber,style=numeric-comp,sorting=none]{biblatex}
\bibliography{cas-refs}

\usepackage{caption}
\usepackage{subcaption}
\usepackage{physics}
\usepackage[hidelinks]{hyperref}
\hypersetup{
  colorlinks   = true, 
  urlcolor     = blue, 
  linkcolor    = blue, 
  citecolor   = red 
}

\usepackage[inner = 20mm, outer = 20mm, top = 20mm, bottom = 20mm]{geometry}
\usepackage[shortlabels]{enumitem}

\theoremstyle{plain}
\newtheorem{lem}{Lemma}[subsection]
\newtheorem{thrm}{Theorem}[subsection]

\theoremstyle{definition}
\newtheorem{defn}{Definition}

\theoremstyle{remark}
\newtheorem*{rmrk}{Remark}

\usepackage{graphicx}
\usepackage[export]{adjustbox}
\usepackage{wrapfig}
\usepackage[shortlabels]{enumitem}   

\newcommand{\etal}{\textit{et al.}\/ }

\newcommand{\R}{\mathbb{R}} 
\newcommand{\Z}{\mathbb{Z}}

\newcommand{\C}{\mathbb{C}}
\newcommand{\F}{\mathbb{F}}
\newcommand{\Id}{\mathds{1}}

\newcommand{\simp}[1]{{( #1 )}}
\newcommand{\pair}[2]{\langle #1,#2 \rangle}
\newcommand{\Mod}[1]{\ (\mathrm{mod}\ #1)}

\newcommand{\argmin}[1]{\underset{#1}{\text{argmin}}}

\title{Counting topological interface modes using simplicial characteristic classes}
\author{N. Bohlsen, I. Y. Dodin, H. Qin}
\date{\today}

\begin{document}

\maketitle

\begin{abstract}
    A computational approach for predicting the number of topological interface modes (TIMs) in hermitian systems using the spectral flow - monopole (SFM) correspondence is presented. The number of TIMs is determined by calculating the Chern number of a complex line bundle of local polarisation vectors over a phase space sphere surrounding a Weyl point. 
    The Chern number is computed by constructing the simplicial first Chern class of a discrete vector bundle on a simplicial mesh. This approach is gauge invariant, derivative free, structure preserving, and robust to noise. The algorithm is shown to reproduce the expected number of TIMs for the case of equatorial fluid waves and the topological Langmuir cyclotron wave. The possibility of using this algorithm to analyse experimental measurements of bulk wave polarisations and predict the associated number of TIMs is explored in a synthetic example. 
\end{abstract}

\tableofcontents
\clearpage

\newpage
\section{Introduction}\label{Section:Introduction}

Topological interface modes (TIMs), which are modes localised to smooth interfaces between separate media, first attracted attention in condensed matter physics and were originally associated with quantum effects and crystal symmetries \cite{asboth2016short,Fidkowski_2011}. However, recent research has demonstrated that they can exist in classical media as well. Examples of TIMs have been reported in geophysics \cite{delplace2017topological,zhu2023topology,perez2025topology,PhysRevResearch.3.043002}, gyroviscous fluid dynamics \cite{PhysRevLett.122.128001}, stellar physics \cite{leclerc2024wave,leclerc2024exceptional}, mechanical systems \cite{PhysRevMaterials.2.124203,huber2016topological,nash2015topological}, plasma physics \cite{qin2023topological,fu2021topological,parker2020topological,fu2024topological}, molecular dynamics \cite{faure2001topological}, and optics \cite{han2024topological} This warrants the attention to the fundamental mathematical theory of TIMs, which is invariant to the physics of the underlying physical media. 

It is understood that, in a given system, the number of TIMs is determined by a topological invariant called Chern number $\text{Ch}_1$ (see Sec. 1.1 for details). Predicting their presence in the spectrum therefore reduces to calculating this Chern number to calculating $\text{Ch}_1$. However, toy models aside, finding the Chern number analytically is difficult. It would be beneficial then to develop a numerical algorithm for calculating $\text{Ch}_1$, especially from noisy simulation or experimental data.

Here, we propose such an algorithm based on the theory of discrete vector bundles and computational algebraic topology. We apply our algorithm to several previously studied  models and show that its predictions agree with the known results. We also demonstrate that our algorithm allows extracting Chern numbers from synthetic experimental data and is noise-tolerant. Notably, our results are not limited to TIM physics and can as well be applied in the context of other topological physics and general applied topology. Our implementation of the algorithms and numerical examples we present is available at the author's Github\cprotect\footnote{ At the URL: \verb!https://github.com/Bohlsen/DiscreteVectorBundles.jl!} as a \verb!Julia! package called \verb!DiscreteVectorBundles.jl!.

This paper is organised as follows. Section 1 reviews the connection between TIMs and the Chern number and then discusses its geometric and topological calculations. Section 2 reviews the elements of the theory of discrete vector bundles and characteristic classes, which we use for developing our algorithm. Section 3 describes and motivates our algorithm to compute the Chern number. Section 4 demonstrates the calculation on relevant examples from the theory of TIMs to show that our numerical scheme agrees with the existing analytic results. Section 5 explains how our method can extract Chern numbers from finite samples of experimental data and tests this on a synthetic example. Section 6 summarises our main results.

\subsection{Spectral flow - monopole correspondence}

For conservative linear systems any wave equation can be rewritten into a Schrodinger equation encoding all topological information into the Hamiltonian operator $\hat{H}$.It is understood that the number of TIMs is related to the topological properties of the Hamiltonian via a Bulk-Boundary (BB) correspondence which equates the number of TIMs to the topology of bulk waves in the media, or more precisely to the topology of the Weyl symbol of the Hamiltonian \cite{asboth2016short,faure2019manifestation,venaille2023ray,delplace2022berry}. However, the classical case is distinct from the condensed matter case due to the lack of global space periodicity in the Hamiltonian which implies that the Brillouin zone is not compact, therefore the standard BB correspondences from the condensed matter literature do not directly apply. In continuum systems without periodicity the relevant BB correspondence is the Spectral Flow - Monopole correspondence.

The Spectral Flow - Monopole (SFM) correspondence concerns solutions to differential equations which take the form of the time dependent Schr\"odinger equation
\begin{equation}
    i\partial_t \ket{\psi(\bm{x},t)} = \hat{H}(\mathbf{x},-i\grad)\ket{\psi(\bm{x},t)}\,,
\end{equation}

where\footnote{Note that throughout this paper we reserve bold symbols for vectors in $\R^3$.} $\mathbf{x}\in \R^3$, $\ket{\psi(\mathbf{x},t)}\in \C^k$ is a $k$-dimensional state vector of physical fields, and $\hat{H}(\mathbf{x},-i\grad)$ is a self-adjoint differential operator. The operator-symbol correspondence assigns to each differential operator, a function on phase space called its Weyl symbol. We notate the correspondence as
\begin{equation}
    \hat{H}(\mathbf{x},-i\grad) \longleftrightarrow H(\mathbf{x},\mathbf{k})\,.
\end{equation}
For detailed discussions of the relevant mathematics we refer to excellent monograph by Folland \cite{folland2016harmonic} or for a more accessible introduction to Tracy \etal \cite{tracy2014ray} and Venaille \etal \cite{venaille2023ray}. 

Specifically, we are interested in eigenmodes, that is with time dependence of $e^{-i\omega t}$, so that the resulting PDE is 

\begin{equation}\label{TISE}
    \omega\ket{\psi} = \hat{H}\ket{\psi}\,.
\end{equation}
Physically relevant equations of the form \eqref{TISE} include eigenstates of quantum systems and standing waves in fluids and plasmas with no dissipation, such as the TIMs discussed above. 

Consider the case where $\hat{H}$ depends upon a free parameter $\lambda$, which we call the \textit{spectral-flow} parameter. In the case of TIMs, the relevant $\lambda$ is the component of the wavevector directed along the boundary, but spectral-flow parameters can also emerge for other reasons. For example, in semiclassical molecular dynamics, the relevant $\lambda$ can be the total rotational angular momentum of a molecule \cite{faure2001topological,faure2019manifestation}. 

Assuming that the spectrum is discrete, we can define a set of curves which specify how the spectrum changes with $\lambda$. An example of this is shown as Figure \ref{fig:DiracSpectrum} which presents the spectrum of the operator\footnote{This Hamiltonian models the 2D Dirac equation with space varying mass. Specifically, it is the local model for the boundary between two regions in which the mass changes sign.}
\begin{equation}\label{DiracHam}
    \hat{H} = \begin{pmatrix}
        x & -i\partial_x -i\lambda\\
        -i\partial_x+i\lambda & -x
    \end{pmatrix}\,.
\end{equation}

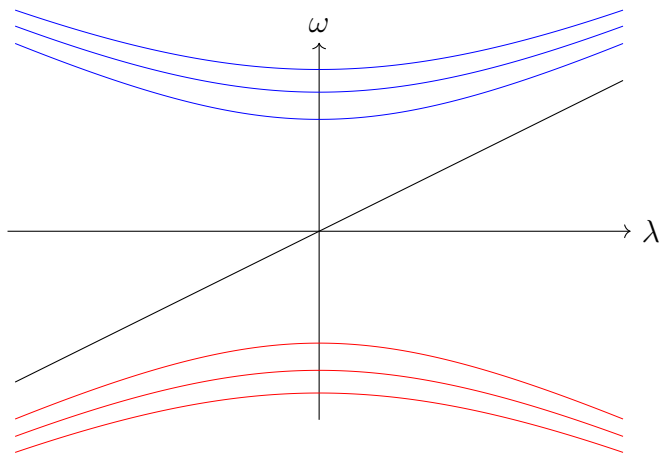
\begin{wrapfigure}{r}{0.5\textwidth}
    \begin{tikzpicture}[domain=-4:4,samples=100,scale = 1]

    \draw[->] (-4.1,0) -- (4.1,0) node[right] {$\lambda$};
    \draw[->] (0,-2.5) -- (0,2.5) node[above] {$\omega$};

    \foreach \n in {1,...,3}{
    \draw[very thin,color=blue] plot(\x,{sqrt{(1.2*\n+1+\x/2*\x/2)}});
    }
    \draw[color=black] plot(\x,{\x/2}) node[below right,font=\small] {};

    \foreach \n in {1,...,3}{\draw[very thin,color=red] plot(\x,{-sqrt{(1.2*\n+1+\x/2*\x/2)}});}
\end{tikzpicture}
    \caption{Spectrum of the 2D Dirac equation with linear mass profile. There is a positive frequency band (blue) a negative frequency band (red) and a spectrally flowing (black) mode which \textit{crosses the gap} and is topologically protected.}
    \label{fig:DiracSpectrum}
\end{wrapfigure}

This spectrum contains a positive frequency and negative frequency \textit{band} separated by a \textit{band gap}. However, there is also a distinguished mode which transits between the two bands as $\lambda$ is varied. This transition is called a \textit{spectral-flow} and for the case of topological boundary modes, one can think of such modes as surface wave which can only propagate in one direction along the boundary. What is notable about these spectrally flowing modes is that they are topologically protected. This means that even reasonably large perturbations to the operator $\hat{H}$ do not destroy the existence of a curve transiting between the two branches. This suggests that the existence of spectral-flow is related to topological properties of $\hat{H}$.

As it turns out \cite{faure2019manifestation,delplace2022berry}, the number of spectrally flowing modes depends upon the existence of degenerate points in the Weyl symbol of $\hat{H}$. More precisely, the SFM correspondence states that the net spectral flow, $\Delta \mathcal{N}$, (i.e. the number of modes crossing the gap as $\lambda$ varies from $-\infty$ to $\infty$) is given by 
\begin{equation}\label{SFMCorrespondence}
    \Delta \mathcal{N} = -\text{Ch}_1\,,
\end{equation}
where $\text{Ch}_1$ refers to the first Chern number of the complex line bundle of $H$ eigenvectors\footnote{We will refer to the $H$ eigenvectors as the \textit{local polarisation} vectors.} over a sphere surrounding a degeneracy point of $H$. Note that the above is only applicable for systems with one inhomogenous dimension and the relevant BB correspondence is more complicated when more dimensions are involved \cite{iwai2023bulk}.

\begin{rmrk}
    To clarify the above, focus on the case of \eqref{DiracHam} and note that $H$ will be a $2\times 2$ matrix valued function. Each band in the spectrum of $\hat{H}$ is associated to a solution \textit{branch} of the local dispersion relation of $H$, that is of $\det (H-\omega) = 0$. This implicitly defines two functions $\omega_\pm(x,k,\lambda)$ which are non-negative and non-positive respectively. Each branch determines a complex line bundle over the space $(x,k,\lambda)$ with fibers taken as the vector space of solutions to $H(x,k,\lambda)\chi^\pm =\omega^\pm(x,k,\lambda)\chi^\pm$. At each $(x,k,\lambda)$ we can choose a normalised local polarisation vector $\chi^\pm(x,k,\lambda)\in\C^2$ which forms a basis\footnote{The choice of basis is not unique and the function $\chi^\pm(x,k,\lambda)$ will not be continuous if the bundle is nontrivial.} for the fiber over $(x,k,\lambda)$. 
    
    However, at any Weyl (degeneracy) point of $H$, that is where $\omega_+=\omega_-$, the local polarisation is undefined and hence the bundles do not exist there. However, we can restrict the bundles to a sphere surrounding the degeneracy point over which they are well defined. The $\text{Ch}_1$ in \eqref{SFMCorrespondence} specifically refers to the Chern number of the $\chi^+$ bundle over this sphere.
\end{rmrk}

\begin{rmrk}
    Generically, if we have a state $\ket{\psi}$ containing $k$ fields then there will exist $k$ non-degenerate bands. We can label the bands by an index $n$ and notate their associated dispersion surfaces and polarisations by $\omega_n(x,k,\lambda)$ and $\chi_n$ respectively. 
    
    There may exist simultaneous spectral flows between any two bands separated by a band gap and this is predicted by a generalised SFM correspondence which states that the net number of modes that spectrally flow into the n'th band is 
    \begin{equation}\label{GeneralisedSFM}
        \Delta\mathcal{N}_n = -\text{Ch}_1^n\,,
    \end{equation}
    where $\text{Ch}_1^n$ refers to the Chern number of the line bundle generated by $\chi_n$ on a sphere surrounding a Weyl point.
\end{rmrk}

\subsection{Geometric calculation of the Chern number}

Within the topological physics literature, the first Chern number is often defined as the integral of the Berry curvature\footnote{Which is defined as $\mathcal{F}_n = d\mathcal{A}_n$, where $\mathcal{A}_n = \frac{i}{2\pi} \chi_n^\dag d\chi_n$ refers to the Berry connection associated to the $n$'th band.} $\text{Ch}_1 = \int \mathcal{F}_n $ \cite{delplace2022berry}. This definition is practical for simple analytic calculations since the eigenvectors $\chi_n$ can often be obtained exactly and hence the Berry curvature calculated and integrated. However, using this definition is inconvenient in numerical analysis because calculations of derivatives from noisy data will introduce substantial error. Therefore while $\text{Ch}_1$ should be exactly an integer the output of a numerical calculation will not be. Furthermore, when dealing with more geometrically complicated domains it is not trivial to construct a discrete approximation to the exterior derivative $d$. The source of the above problems is the adoption of a geometric definition of the Chern number but this is actually unnecessary as we can instead adopt a purely topological definition which is completely free of derivatives, as will be discussed in Section \ref{sec:TopologicalInfo}.

Other authors have investigated how to calculate the first Chern number of a complex line bundle relying upon the usual geometric definition. Here, we briefly review the existing approaches and discuss their limitations. 

The interpretation of the $\text{Ch}_1$ as the integral of the Berry curvature comes from the mathematics of gauge theory, which studies the geometry of connections on vector bundles \cite{bleecker2005gauge}. One approach to calculating $\text{Ch}_1$ then comes from constructing a discrete gauge theory. Due to research in Lattice Gauge Theory this is a relatively well understood problem. Fukui \etal developed a numerical method which can compute $\text{Ch}_1$ for complex line bundles over the Brillouin zone of a topological insulator \cite{fukui2005chern}. Their approach works by discretizing the Brillouin zone, which is topologically a torus $S^1\times S^1$, into a regular grid, and then constructing a discrete connection and curvature on that grid. This works particularly well for Brillouin zones because their geometry is known in advance. For the case of more general triangulations of base manifolds, or vector bundles derived from data the same approach is not simply applicable. Similar methods have been applied by several other authors \cite{zhao2020first,kudo2019many,wang2020universal,knoppel2016complex}. However, all of the discrete gauge theory methods suffer from the same issue as just directly integrating a numerical Berry curvature. Since $\text{Ch}_1$ must be an integer, we prefer approaches where the integer nature of $\text{Ch}_1$ is exactly enforced by the calculation. This can be achieved and has been investigated in the past. For example, one approach is to connect the Chern number to a relevant Bott index and then calculate that instead such as attempted by Lin \etal  \cite{toniolo2022bott,lin2023calculations}. However, this is really only practical in the case of topological insulators as it relies upon the toroidal topology of the Brillouin zone. Similar ideas have been explored by Tarnowski \etal who experimentally investigated the connection between the $\text{Ch}_1$ and an integer linking number of curves in a time dependent momentum space \cite{tarnowski2019measuring}. However, again, their approach is very specific to the case the Brillouin zone of a topological insulator. Finally, an alternative approach would be to compute $\text{Ch}_1$ via intersection theory, observing that it counts the number of signed intersections between a generic section of a line bundle and its zero section. This is a standard calculation analytically, for example see the calculation of Faure \cite{faure2001topological}, but to our current knowledge has not been implemented numerically as a general tool.

\subsection{Topological calculation of the Chern number}\label{sec:TopologicalInfo}

Suppose we want to count the expected number of TIMs for a system with $k$ fields and a Hamiltonian $\hat{H}$. Take that the Weyl symbol of $\hat{H}$ can be written as a $k\times k$ matrix $H(\bm{\lambda})$ where $\bm{\lambda}\in \R^3$ represents a general choice of coordinates for the 3D parameter space, of which one coordinate direction must be the spectral flow paramter. Also suppose that we have precalculated the dispersion surfaces and chosen a branch $\omega_n(\bm{\lambda})$ of which we will compute the Chern number\footnote{This is not strictly necessary but it is a useful simplifying assumption for now.}. 

We propose an algorithm for calculating $\text{Ch}_1$ using only a sample of local polarisation vectors taken from the vertexes of a triangular mesh surrounding a Weyl point. The mesh will need to be structured as a \textit{simplicial complex}. The following remark provides a brief description of simplicial complexes and establishes our notation.

\begin{rmrk}
    A simplicial complex $K$ is an assembly of simplexes (generalised triangles) of different dimensions. In a computational physics context we think of $K$ as a triangular mesh approximating our manifold \cite{edelsbrunner2010computational,ghrist2014elementary}. The notation we adopt to describe simplicial complexes is as follows.

    We indicate the subset\footnote{Note that this notation differs from the usual convention in topology where $K^p$ refers to the complete p-skeleton on $K$.} of simplices of dimension $p$ as $K^p\subset K$, so $K^0$ is then the set of vertices of $K$ which we will assume are finite and given in a explicit order. We notate each vertex (0-simplex) as $(i)\in K^0$ where $i$ is the index of the vertex. The higher simplices are each then labelled by their vertices. For example $(ij)$ is the edge (1-simplex) which connects $(i)$ to $(j)$, and $(ijk)$ is the triangle with vertices $(i),(j),(k)$. 
    
    The assumed ordering on higher simplexes is lexicographical, so $i<j<k$ for $(ijk)$, and this specifies the assumed orientation on each simplex. 
\end{rmrk}

Our proposed algorithm can be summarised into the following list of 8 steps.

\begin{enumerate}
    \item Choose a sample of $N$ points $\{\bm{\lambda}_i\}_{i= 1}^N$ from a sphere surrounding a Weyl point.
    \item Use these points as the vertices to generate a spherical triangular mesh $K$ and its boundary matrices\footnote{The definition of the boundary matrices is provided in Appendix \ref{AppendixA:TDA}.} $B_p$ for this mesh.
    \item At each vertex $(i)\in K^0$ solve $H(\bm{\lambda}_i)\chi_{(i)} = \omega_n(\bm{\lambda}_i)\chi_{(i)}$ for the local polarisation $\chi_{(i)}$. 
    
    \item Make an array $\Phi = \{\Phi_{(i)}\}_{(i)\in K^0}$ of $\R$ valued $k\times 2$ projection matrices defined as $\Phi_{(i)} = [\gamma(\chi_{(i)}),\gamma(J\chi_{(i)})]$ where $J$ refers to multiplication by the complex unit and $\gamma$ is a linear map which is defined in Section \ref{sec:FirstChernClass}.

    \item For each edge $(ij)\in K^1$ calculate the SVD\footnote{Singular Value Decomposition} of the matrix $\Phi_{(i)}^\intercal\Phi_{(j)} = U_{(ij)}S_{(ij)}V_{(ij)}^\intercal$. Then construct an array of angles $\Theta = \{\Theta_{(ij)}\}_{(ij)\in K^1}$ defined by $\Theta_{(ij)} = \arctan\left(\frac{[U_{(ij)}V_{(ij)}^T]_{21}}{[U_{(ij)}V_{(ij)}^T]_{11}}\right)$.
    
    \item Construct a $\Z$ valued array $e = \{e_{(ijk)}\}_{(ijk)\in K^2}$ defined by $e_{(ijk)} = \text{round}\left(\frac{\Theta_{(ij)}-\Theta_{(ik)}+\Theta_{(jk)}}{2\pi}\right)$.
    \item Solve $B_2\mu=0$ for a $\Z$ valued null vector\footnote{There is an important subtlety here since the equation $B_2\mu = 0$ does not have unique solutions. Determining the correct solution will be discussed in \ref{sec:OrientingK}.}. Store it in an array $\mu = \{\mu_{(ijk)}\}_{(ijk)\in K^2}$.
    \item Compute the sum $\text{Ch}_1^n = \sum_{(ijk)\in K^2} \mu_{(ijk)}e_{(ijk)}$.
\end{enumerate}

The $\text{Ch}_1^n$ output by the above will be the Chern number\footnote{Guaranteeing that our algorithm correctly computes the Chern number requires some assumptions on both the geometry of the sample of points and on the Lipshitz continuity of the underlying vector bundle. This will be discussed in Section \ref{Section:EulerClass}.} of the polarisation bundle for the $\omega_n$ branch and hence the net number of TIMs which enter the $n$'th band the wavevector along the interface is varied. To readers who are used to the geometric approach to calculating $\text{Ch}_1$ it may seem somewhat strange that the above algorithm would compute it at all. Explaining \textit{why} this algorithm calculates $\text{Ch}_1$ requires a discussion of the topological theory of \textit{simplicial characteristic classes}.

\subsubsection{Topological motivation}

The topological definition of a Chern number follows from the theory of characteristic classes of vector bundles. If $E^m\rightarrow M$ is a rank-$m$ vector bundle over a manifold $M$, then its characteristic classes are an assignment, to the bundle, of cohomology classes of $M$ which, in some sense, measure a property of the bundle. Characteristic classes are topological invariants, that is, they depend only upon the topological character of the vector bundle and not on any geometric structure. This implies that they can be defined upon more abstract mathematical spaces than just manifolds, and specifically can be constructed on simplicial complexes, on which their calculation is both purely combinatorial and robust to noise.  

\begin{rmrk}
       In the above are implicitly referring to the \textit{simplicial cohomology} instead of the more common, in physics at least, \textit{singular} or \textit{De Rham} cohomologies. For readers who are not familiar with simplicial (co)homology, a brief review is presented in Appendix \ref{AppendixA:TDA} which also defines our chosen notation for simplicial cohomology. 
\end{rmrk}

The primary characteristic class we need in order to define $\text{Ch}_1$ is the first Chern class $[c_1]\in H^2(M,\Z)$. Later, in Section \ref{Section:EulerClass}, we will also need the first Stiefel--Whitney class $[W_1] \in H^1(M,\Z_2)$, and the Euler class $[e]\in H^{m}(M,\Z)$. Each of these measures a formally different property of the bundle: $[W_1]$ measures the orientability of the bundle, $[e]$ measures whether there exists a nowhere vanishing section, and $[c_1]$ measures whether it is possible to assign a consistent phase throughout a complex vector bundle \cite{milnor1974characteristic}. Full definitions of these classes will be provided in Section \ref{subsec:SWAndEuler}.

The core piece of homology theory we need to define the Chern number is the construction of the \textit{fundamental class} of a manifold. Suppose $M$ is a connected, compact, orientable, $n$-dimensional manifold, then it turns out that the top dimensional homology group $H_n(M,\Z)$ must be isomorphic to the integers, that is $H_n(M,\Z) = \Z$ \cite{Hatcher}. Establishing an orientation on $M$ amounts to choosing a generator for $H_n(M,\Z)$, by which we mean finding a cycle $\mu_M\in \ker B_n$ such that $H_n(M,\Z) = \{ \alpha[\mu_M],\alpha\in \Z\}$. $[\mu_M]\in H_n(M,\Z)$ is called the \textit{fundamental class} of $M$. Note that there are two possible orientations we could choose for $M$ and hence two valid choices for the fundamental class. Determining the geometrically correct fundamental class will be discussed in Section \ref{sec:OrientingK}.

Equipped with $[c_1]$ and $[\mu_M]$ we can now provide a topological definition for the Chern number of a vector bundle over a manifold. For the purposes of counting TIMs we only need 2-sphere, $S^2$, base manifolds and so we will restrict to the case where $M$ is a $2$-dimensional surface.
\begin{defn}
    Let $E^m\rightarrow M$ be rank-$m$ complex vector bundle over a base manifold $M$ which is connected, compact, orientable and $2$-dimensional. The first Chern number $\text{Ch}_1$ is the result of pairing $[c_1]\in H^2(M,\Z)$ with $[\mu_M] \in H_2(M,\Z)$, the fundamental class of $M$. 
    \begin{equation}\label{TopoC1}
        \text{Ch}_1 = \pair{[c_1]}{[\mu_M]}\,.
    \end{equation}
\end{defn}

\begin{rmrk}
    Note that this topological definition for $\text{Ch}_1$ agrees with the geometric definition constructed using Berry curvature, or more generally Chern-Weil theory \cite{milnor1974characteristic,nakahara2018geometry}.
\end{rmrk}

\subsubsection{General scheme}

The core idea behind our algorithm to calculate $\text{Ch}_1$ is to translate \eqref{TopoC1} onto a simplicial triangulation $K$ of $M$. Then, using this triangulation, $\mu_M$ and $c_1$ can be represented by vectors in the simplicial chain and cochain basis for which calculating the pairing $\pair{[\mu_M]}{[c_1]}$ just requires taking a dot product. This calculation is necessarily derivative free, purely $\Z$ valued, and robust to noise, since it depends only on the homology and cohomology class of $\mu_M$ and $c_1$ which are robust properties of the associated chain and cochain. 

Therefore, at a topological level our algorithm can be described as follows:
\begin{enumerate}
    \item Generate a simplicial triangulation $K$ of the base manifold $M$.
    \item Construct the complex line bundle of local polarisation vectors over $K$.
    \item Obtain a simplicial cochain $c_1\in C^2(K,\Z)$ representing the first Chern class $[c_1]\in H^2(K,\Z)$ of the bundle.
    \item Obtain a simplicial chain $\mu_K \in C_2(K,\Z)$ representing the oriented fundamental class of $K$ $[\mu_K]\in H_2(K,\Z)$.
    \item Calculate the pairing $\text{Ch}_1 = \pair{[c_1]}{[\mu_K]} = c_1^T\mu_K$.
\end{enumerate}

Generating the simplicial triangulation $K$ and calculating $\mu_K$ are both relatively simple tasks with well developed theory and existing algorithms to choose from \cite{edelsbrunner2010computational}. However, constructing a vector bundle over a simplicial complex and computing $c_1$ are much less well developed problems with little existing theory to apply. The main obstruction is that the usual definition of a vector bundle, familiar to most people who have taken a class on differential geometry, is tailored to bundles over manifolds and is challenging to adapt to simplicial complexes since they are discrete spaces rather than continuous ones. The solve this we can define the related notion of an \textit{approximate discrete vector bundle} over $K$ and their associated simplicial characteristic classes. In Section \ref{Section:DiscreteVectorBundles} we will discuss discrete vector bundles and the associated characteristic classes in detail.

\subsubsection{Existing literature on simplicial characteristic classes}\label{sec:TopoLitReview}

Separately to the topological physics research program, there has been an ongoing research program in computational geometry and topology which has looked at computing characteristic classes of vector bundles more generally and developed algorithms to do so. Tinarrage presented a general scheme for the calculation of the Stiefel--Whitney classes of real vector bundles \cite{tinarrage2022computing}. In principle, their method could be extended to calculate the $c_1$ class we need to obtain but Tinarrage's approach requires constructing explicit triangulations of Grassmanians which they determined was computationally intractable. A related approach was developed independently by Ren \cite{ren2025persistent}. The work of Tinarrage and Ren has been recently extended by Gang \cite{gang2025persistent} who presented a method which practically calculates higher Stiefel--Whitney classes, for datasets which approximate smooth manifolds, using Wu's identities \cite{milnor1974characteristic}. Finally, and for our purposes here most relevantly, Scocolla and Perea \cite{scoccola} presented a method for the computation of Euler classes of a rank-2 real vector bundle using discrete \v{C}ech cocycles. Their method works well with finite samples of vector bundles with the base space embedded in $\R^N$, and hence is applicable for counting TIMs. The algorithm we present in Section \ref{Section:EulerClass} to calculate $c_1$ is essentially a direct application of the methods of Scoccola and Perea combined with the observation that the first Chern class of a complex line bundle equals the Euler class of its realisation. Therefore, by transforming a finite sample of a complex line bundle into a finite sample of a real rank 2 vector bundle we can apply Scoccola and Perea's algorithm to calculate the Euler class and obtain $c_1$ as a result.

\subsection{Notational conventions}

To simplify later discussion we will briefly discuss some notational choices we adopt. 

Consider the space of orthonormal bases of $k$ dimensional planes in $\R^N$. We can describe refer to this space as $V(k,N;\R)$ and represent it by $N\times k$ real matrices with orthonormal columns. This space comes naturally embedded in $V(k;\R)$, the space of $\infty \times k $ matrices with orthonormal columns. We refer to these two spaces generically as the \textit{Stiefel matrices}. When referring to the Stiefel matrices we will usually assume real entries, but the complex Stiefel matrices will also be required. When there is a chance of confusion the real and complex cases are notated as $V(k;\R)$ and $V(k;\C)$ respectively. Also, we adopt the standard $O(k)$ and $SO(k)$ notation for the spaces of orthogonal and special orthogonal $k\times k$ matrices. 

We will regularly attach numerical data, such as matrices, to the vertices, edges, and triangles of a simplicial mesh $K$. To indicate this we use the simplex label as a subscript, as was the case in Section \ref{sec:TopologicalInfo}. For example, if we have a collection of vectors $\chi$ each of which is assigned to an edge then we will notate each element as $\chi_{\simp{ij}}$ which refers to the vector attached to the specific edge $\simp{ij}$. Likewise $\Gamma_{\simp{lmn}}$ refers to a mathematical object of some kind attached to the triangle $\simp{lmn}$. This convention is adopted to distinguish objects attached to cells from expressions like $[M]_{ij}$ which refers to the row $i$, column $j$ element of a matrix $M$.

The following are mathematical conventions, rather than notational conventions we adopt: All spaces are assumed to be topological spaces, usually manifolds. All base manifolds of vector bundles are assumed to be compact, and all other spaces are assumed to be paracompact. All vector spaces are assumed to be Euclidean, equipped with an inner product, and hence all bases can be assumed given to us as orthonormal bases since any that are not orthonormal can be made so by the Gram-Schmidt procedure.

\section{Approximate and discrete vector bundles}\label{Section:DiscreteVectorBundles}

In this section we discuss the relevant mathematical background needed to explain the algorithm for counting TIMs which was presented in Section \ref{sec:TopologicalInfo}. We will begin by explaining how we represent a vector bundle which is obtained from numerical or experimental data. We then develop the notion of an approximate discrete vector bundle over a simplicial complex. Finally, we will define the Stiefel--Whitney, Euler, and Chern classes and explain the relationship between them.

\subsection{Representing vector bundles obtained from data}\label{subsec:vectorbund}

Traditionally, a rank-$m$ vector bundle $E^m\rightarrow M$ is specified by a projection map $\pi$, an open cover\footnote{Here $I$ refers to an index set which we leave unspecified.} $\mathcal{U} = \{U_i\}_{i\in I}$ of $M$, and a homeomorphism $\phi_i:\pi^{-1}(U_i)\rightarrow U_i\times \R^m$ for each open set $U_i$. However, when dealing with vector bundles obtained from numerical or experimental data this definition is not natural. The problem is that the fibers $\pi^{-1}(p)$ are formally abstract and do not come embedded in an ambient space. Most processes which assign vector spaces locally will specify them as subspaces of a given ambient vector space, usually with a given basis. For example, for the local polarisation bundles of a Weyl symbol $H$, the fibers are embedded in an ambient copy of $\C^k$. 
Conveniently, it is possible to take any vector bundle and transform it into a description in terms of framed fibers embedded in an ambient space. To do this we need the following lemma.
\begin{lem}[Lemma 5.3 \cite{milnor1974characteristic}]\label{lem:Milnor5.3}
    For any vector bundle over a paracompact base space $E^m\rightarrow M$ there exists a map $f:E^m\rightarrow \R^N$, and an $N$ large enough, such that on fibers $f$ is a linear injection.
\end{lem}
Equipped with the Lemma \ref{lem:Milnor5.3}, we consider the maps $f\circ\phi_i^{-1}:U_i\times \R^m\rightarrow \R^N$. Fixing the standard basis on $\R^m$ and $\R^N$ and noting that $\phi_i$ and $f$ are both linear and injective on fibers we can represent this map by a matrix, with column rank-$m$, which varies over $U_i$. Performing the Gram--Schmidt process allows us to ensure this matrix has orthonormal columns and so we can always represent $f\circ\phi_i^{-1}$ by a map $\Phi_i:U_i\rightarrow V(m,N;\R)\subset V(m;\R)$. This map assigns to each point in $U_i$ an orthonormal basis, the columns of $\Phi_i$, whose span represents the fiber embedded in $\R^N$. We refer to this basis as a \textit{frame} over $U_i$ and we will refer to the collection $\Phi$ as a \textit{local trivialisation} of the vector bundle\footnote{This is somewhat non-canonical nomenclature. }. 

The construction above assigns to each vector bundle a local trivialisation $\Phi$ but not completely freely. The vector bundle structure non-trivially constrains how the maps behave on overlaps in $\mathcal{U}$. Consider the overlap of two open sets $U_i\cap U_j$, on which it follows that $\Phi_i$ and $\Phi_j$ are related by 
\begin{equation}\label{Witness}
    \Phi_j = [f\circ \phi_j^{-1}] = [f\circ \phi_i^{-1}\circ \phi_i\circ \phi_j^{-1}] = [f\circ \phi_i^{-1}][\phi_i\circ \phi_j^{-1}]  = \Phi_i\Omega_{ij}\,,
\end{equation}
where $[\cdot]$ indicates taking a matrix representation of the map and $\Omega_{ij}:U_i\cap U_j\rightarrow O(m)$ is the matrix representation of a transition function $\phi_i\circ\phi_j^{-1}$. 

The set of maps $\Omega_{ij}$ is called a \textit{witness} for the local trivialisation $\Phi$. Using this concept of a witness we can construct a computationally practical definition of a vector bundle. 

\begin{defn}[Local Trivialisation]\label{def:Local Trivialisation}
    Given an open cover $\mathcal{U}$ of $M$, a local trivialisation of a rank-$m$ vector bundle is an assignment of a map $\Phi_i:U_i\rightarrow V(k;\R)$ for each $U_i\in \mathcal{U}$, such that for each non-empty intersection $U_i\cap U_j$ there exists a map $\Omega_{ij}:
    U_i\cap U_j\rightarrow O(m)$ for which 
    \begin{equation}\label{WitnessCondition}
        \Phi_i \Omega_{ij}=\Phi_j\,.
    \end{equation}
\end{defn}

\begin{rmrk}
    Definition \ref{def:Local Trivialisation} is equivalent to the standard definition of a vector bundle as it is possible to use the maps $\Phi$ to construct a unique bundle in the traditional sense \cite{lee2003smooth}. 
\end{rmrk}

\begin{rmrk}
    As an example of how such a local trivialisation is obtained from data, consider our TIM counting problem from above. Recall that we seek to solve the equations $H(\bm{\lambda})\chi(\bm{\lambda}) = \omega(\bm{\lambda})\chi(\bm{\lambda})$ and $\chi(\bm{\lambda})^\dag \chi(\bm{\lambda})=1$ for the normalised local polarisation vectors $\chi\in \C^k$ over all the points $\bm{\lambda}\in S^2$ surrounding a Weyl point. A smooth global solution will not necessarily exist but local solutions always do. That is, around each point $\bm{\lambda}$ there exists an open set $U_{\bm{\lambda}}\subset S^2$ and a local function $\chi_{\bm{\lambda}}:U_{\bm{\lambda}}\rightarrow \C^k$ such that $H(\bm{\lambda}')\chi_{\bm{\lambda}}(\bm{\lambda}') = \omega(\bm{\lambda}')\chi_{\bm{\lambda}}(\bm{\lambda}')$ for all $\bm{\lambda}'\in U_{\bm{\lambda}}$. We therefore have an infinite open cover $\{U_{\bm{\lambda}}\}_{\bm{\lambda}\in S^2}$ which we can restrict to some finite subcover $\mathcal{U} =\{U_i:= U_{\bm{\lambda_i}}\}_{i\in I}$. To each set $U_i$, we can attach a Stiefel matrix valued map $\Phi_i: U_i \rightarrow V(2,k;\R)\subset V(2;\R)$ given by $\Phi_i(\bm{\lambda}')= [\gamma(\chi_{\bm{\lambda}_i}(\bm{\lambda}')),\gamma(J\chi_{\bm{\lambda}_i}(\bm{\lambda}'))]$, where the $\gamma$ and $J$ maps are defined in Section \ref{sec:FirstChernClass}. The open cover $\mathcal{U}$ of $S^2$ and set of maps $\Phi = \{\Phi_i\}$ constitute the data for a local trivialisation of the rank-$2$ local polarisation bundle over $S^2$.

    Note that obtaining the witness $\Omega$ associated to this trivialisation is a straightforward exercise after noting that on any overlap $U_i\cap U_j$ the functions $\chi_{\bm{\lambda}_i}$ and $\chi_{\bm{\lambda}_j}$ must agree up to a complex phase factor, since the local polarisations are unique up to this phase. Therefore, there exists smooth functions $\theta_{ij}:U_i\cap U_j \rightarrow \R$ such that $e^{i\theta_{ij}}\chi_{\bm{\lambda}_i} = \chi_{\bm{\lambda}_j}$ from which the $\Omega_{ij}\in O(2)$ matrices can be constructed.
\end{rmrk}

In section \ref{Section:EulerClass}, we will need to deal with the problem of orientation on vector bundles. Choosing a frame $\Phi_i$ over each $U_i$ specifies an implicit order on the basis vectors of each fiber by the ordering of the columns of $\Phi_i$. This amounts to a choice of local orientation on the vector bundle. A \textit{global orientation} is then a choice of local trivialisation for which the frames are consistently oriented. We can formalise this idea as Definition \ref{def:orientation}
\begin{defn}\label{def:orientation}
    A vector bundle specified by a local trivialisation $\Phi$ is oriented if the transition functions $\Omega$ are orientation preserving. That is $\det(\Omega_{ij}) >0$ for all $U_i\cap U_j \neq 0$. It is orientable if we can choose an equivalent local trivialisation such that the above holds.
\end{defn}

\subsection{Approximate vector bundles}\label{subsec:approxvectorbund}

If we want to study vector bundles numerically, we will eventually need to deal with the problem of redundancy to noise. Suppose we perform some set of measurements which produce the maps $\Phi$ assigned to our open sets. Then for any point in an overlap $p\in U_i\cap U_j$ we would expect that $\Im\Phi_i(p) = \Im\Phi_j(p)$ since these maps just describe different bases over the same fiber. Any amount of noise will cause this result to fail slightly and result in $\Im\Phi_i(p)$ and $\Im\Phi_j(p)$ referring to slightly different hyperplanes. As a result we should not expect to satisfy \eqref{Witness} exactly, but we may satisfy it approximately. That is, while there may not exist an $\Omega_{ij}$ such that \eqref{Witness} holds, we may have an $\Omega_{ij}$ for which $\Phi_i\Omega_{ij} \approx \Phi_j$. Placing a norm on the matrix subspace $V(k;\R)\subset \R^{\infty\times k}$ we can define an \textit{approximate vector bundle} by requiring that this approximation be uniformly good across the whole cover. 
\begin{defn}[Approximate Local Trivialisation]\label{def:Approximate Local Trivialisation}
    Given $\epsilon>0$, an $\epsilon$-approximate local trivialisation consists of identical data to Definition \ref{def:Local Trivialisation} but with the condition \eqref{WitnessCondition} weakened to
    \begin{equation}
        ||\Phi_i \Omega_{ij}-\Phi_j||<\epsilon\,.
    \end{equation}
\end{defn}

\begin{rmrk}
Definition \ref{def:Approximate Local Trivialisation} was introduced by Scoccola and Perea \cite{scoccola}. They extended a result of Tinarrage \cite{tinarrage2022computing} to show that as long as $\epsilon$ is sufficiently small an approximate vector bundle projects onto a true vector bundle over the base space and that this association is stable with respect to perturbation of the approximate local trivialisation $\Phi$. The exact constraints on $\epsilon$ are discussed in their original paper and we refer interested readers there for details. 
\end{rmrk}

The above discussion focused on the local trivialisations $\Phi$ of a vector bundle. The transition maps $\Omega$ only play the role of witnesses under this interpretation of the theory. There exists a related alternative theory, which we call the \v{C}ech theory, in which a vector bundle can be entirely specified by its transition maps. Considering a non-empty overlap between three open sets $U_i\cap U_j\cap U_k$, performing a similar calculation to \eqref{Witness} and using the fact that each $\Phi_i$ is an isomorphism onto its image yields 
\begin{equation}\label{CocycleCondition}
    \Omega_{ik} = \Omega_{ij}\Omega_{jk}\,.
\end{equation}
This is known as the \textit{cocycle condition} and it is the core constraint on the transition functions of a vector bundle. A set of maps $\Omega_{ij}:U_i\cap  U_j\rightarrow O(k)$ satisfying \eqref{CocycleCondition} exactly specifies a vector bundle\footnote{This result is sometimes called the vector bundle reconstruction theorem \cite{lee2003smooth}.}. 

It will be helpful later to introduce some nomenclature related to the \v{C}ech theory of vector bundles. An assignment of maps $\Theta_i:U_i\rightarrow O(k)$ for each $U_i$ in the open cover is referred to as a \v{C}ech 0-cochain, the space of which is a group we notate as $\check{C}^0(\mathcal{U},O(k))$. If we assign maps to the overlaps $\Omega_{ij}:U_i\cap  U_j\rightarrow O(k)$ we call this a \v{C}ech 1-cochain, the space of which we refer to as $\check{C}^1(\mathcal{U},O(k))$, and further if $\Omega \in \check{C}^1(\mathcal{U},O(k))$ satisfies \eqref{CocycleCondition} then $\Omega$ will be called a \v{C}ech 1-cocycle \cite{nakahara2018geometry}. Note that the orthogonal group $O(k)$ here can be replaced by any group, although we will only need $O(k)$ and $SO(k)$ as they are the structure groups relevant to our vector bundles.

\begin{rmrk}
    Note that, the witness $\Omega$ associated to an approximate local trivialisation $\Phi$ will not exactly satisfy \eqref{CocycleCondition} but instead will satisfy the weaker constraint $||\Omega_{ik}-\Omega_{ij}\Omega_{jk}||<3\epsilon$ and so will constitute only an approximate \v{C}ech cocycle. The factor of $3$ here follows from repeated application of the triangle inequality and is proven as Lemma 4.9 in \cite{scoccola}.
\end{rmrk}

\subsection{Discrete vector bundles}

The above discussion of approximate vector bundles was entirely continuous and is not directly implementable as a discrete numerical theory. This is a common issue in applied topology and is solved by observing that most of the information given in $\mathcal{U}$ and the local trivialisation $\Phi$ is redundant, in the sense that we can reduce it to a purely combinatorial description of the bundle from which we can reconstruct our original continuous bundle. 
Given an open cover $\mathcal{U} = \{U_i\}_{i\in I}$ we can construct a simplicial complex, called the \textit{Nerve} of the cover, $\mathcal{N}(\mathcal{U})$, which encodes the local overlaps between the open sets in the cover \cite{ghrist2014elementary}. $\mathcal{N}(\mathcal{U})$ is an abstract simplicial complex and is constructed as follows: include all singleton elements $(i)\in I$, then all pairs (edges) $(ij)$ for which $U_i\cap U_j\neq 0$, then all triplets (triangles) $(ijk)$ for which $U_i\cap U_j\cap U_k\neq 0$ and so on. As an example, Figure \ref{fig:Nerve} shows a collection of four open sets and its nerve.

\begin{figure}
    \centering
    \begin{tikzpicture}[scale = 0.8]
    \draw[smooth cycle,tension=1] plot coordinates{(-4,0) (-2,2) (0.5,-0.5) (-2,1)};
    \coordinate (Ui) at (-2,1);
    \node[above right] (i) at (Ui) {$U_i$};

    \draw[smooth cycle,tension=1] plot coordinates{(-5,-0) (-3,0.5) (-1,-1.5) (-4,-2)};
    \coordinate (Uj) at (-4,-1);
    \node[right] (j) at (Uj) {$U_j$};

     \draw[smooth cycle,tension=1] plot coordinates{(-2,-1) (0,0.5) (1,0) (0,-1) (-0.5,-1.5) (-1.5, -2)};
    \coordinate (Uk) at (0,-1);
    \node[above left] (k) at (Uk) {$U_k$};

    \draw[smooth cycle,tension=1] plot coordinates{(-4,-3) (-2,-1.5) (-1,-2) (0,-3.5)};
    \coordinate (Ul) at (-1,-3);
    \node (l) at (Ul) {$U_l$};

    \path[->,black,dashed,thick] (2,-1) edge (4,-1);

    \draw[thick,black]  (6.5,1) -- (7.5,-1) -- (5.5,-1) -- (6.5,1);
    \filldraw[thick, black, fill = gray!30] (5.5,-1) -- (6.5,-3) -- (7.5,-1) -- (5.5,-1);
            
    \node[above] (i2) at (6.5,1) {$\simp{i}$};
    \node[right] (l2) at (7.5,-1) {$\simp{k}$};
    \node[left] (j2) at (5.5,-1) {$\simp{j}$};
    \node[below] (l2) at (6.5,-3) {$\simp{l}$};
    
    \end{tikzpicture}
    \caption{Diagram of a cover $\mathcal{U}$ and its nerve $\mathcal{N}(\mathcal{U})$. The cell $\simp{ijk}$ is not included in the complex, but $\simp{jkl}$ is, as indicated by the shading.}
    \label{fig:Nerve}
\end{figure}
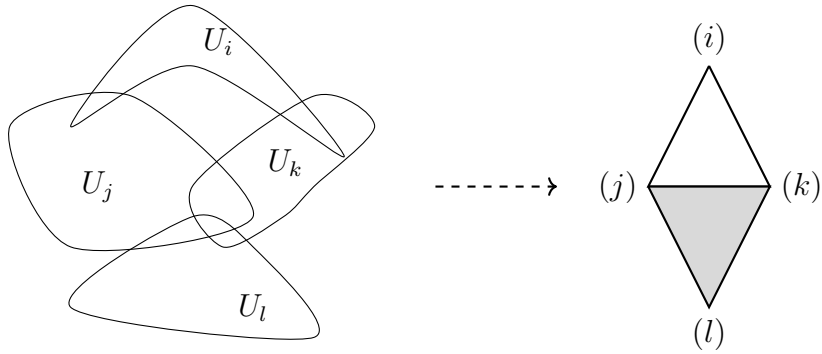

A key result in combinatorial topology, called \textit{the Nerve lemma}, asserts that for a collection of open contractible sets $\mathcal{U}$, for which all of the non-empty intersections of open sets are contractible, $\mathcal{N}(\mathcal{U})$ is homotopy equivalent to $\cup_i U_i$ \cite{ghrist2014elementary}. So, for the case of a good open cover\footnote{A \textit{good cover} is one for which all of the open sets are contractible and have contractible intersections.} $\mathcal{U}$ of a space $M$ we conclude that $\mathcal{N}(\mathcal{U})$ is homotopy equivalent to $M$. Therefore, we lose no topological information by forgetting about the open sets and instead working directly with the nerve. This is helpful because $\mathcal{N}(\mathcal{U})$ is a finite simplicial complex and so can be described with only a finite amount of information, making it amenable to storage on a computer. Note that the example nerve in Figure \ref{fig:Nerve} is clearly a deformation retraction of the open sets.

The redundancy implied by the Nerve lemma suggests that we if equip $\mathcal{N}(\mathcal{U})$ with some notion of a \textit{discrete vector bundle} then we should be able to reconstruct an approximate vector bundle over $\cup_i U_i$. We can then perform all of our topological calculations in the discrete realm and be guaranteed that we are correctly identifying the associated continuous object. 

This idea motivates the following definition for a local trivialisation of an approximate discrete vector bundle \cite{scoccola}.

\begin{defn}\label{DiscreteVectorBundle}
    A discrete $\epsilon$-approximate local trivialisation over a simplicial complex $K$ is an assignment of a Stiefel matrix $\Phi_{(i)}\in V(k;\R)$ to each vertex $(i)\in K^0$ such that for every edge $\simp{ij}\in K^1$ there exists an $\text{O}(k)$ matrix $\Omega_{(ij)}$ for which $||\Phi_{(i)}\Omega_{(ij)}-\Phi_{(j)}|| < \epsilon$.
\end{defn}

\begin{rmrk}
    Within the TIM counting problem the assignment of a specific local polarisation vector $\chi_{\simp{i}}\in \C^k$ to each vertex $\simp{i}\in K^0$ almost defines an approximate discrete local trivialisation over the mesh $K$. However, the vectors $\chi_{\simp{i}}$ are complex valued rather than real valued. To make a true discrete vector bundle we need to translate this complex line bundle into a real plane bundle as is discussed in Section \ref{sec:FirstChernClass}.
\end{rmrk}

A similar definition can also be made for a discrete approximate \v{C}ech cocycle $\Omega \in \check{C}^1(K,O(k))$, which is an assignment of an $O(k)$ rotation matrix to each edge $\simp{ij}\in K^1$ such that the for all triangles $\simp{ijk}\in K^2$ we have a bounded error $||\Omega_{\simp{ij}}\Omega_{\simp{jk}}-\Omega_{\simp{ik}}||<\epsilon$. Note that, we will generically use the notation $\check{C}^p(K,O(k))$ to refer to the set of assignments of $O(k)$ matrices to the $p$-simplices of $K$, and call such an assignment a \v{C}ech $p$-cochain. For example, $\Theta\in \check{C}^0(K,O(k))$ refers to an assignment of an $O(k)$ matrix $\Theta_{\simp{i}}$ for each vertex $\simp{i}\in K^0$.

\begin{rmrk}
    Extending a discrete approximate vector bundle over $\mathcal{N}(\mathcal{U})$ to the full space $\cup_i U_i$ is straightforward as we can just take each map in the continuous local trivialisation to be constantly $\Phi_{(i)}$. Scocolla and Perea proved a series of consistency and reconstruction theorems which show that as long as $\epsilon$ is small enough, the discrete vector bundle extends to the correct continuous vector bundle and further that the characteristic classes computed from it are representative of the true characteristic classes of the underlying continuous bundle. Details on the relevant technical constraints can again be found in the original publication, but we will largely ignore these details here \cite{scoccola}. 
\end{rmrk} 

\subsection{Stiefel--Whitney and Euler classes}\label{subsec:SWAndEuler}

The algorithms we present below for computing the first Chern number of a complex line bundle descend from the computation of the Stiefel--Whitney and Euler classes of real plane bundles. For reference we review their definition here. Note that in this subsection we suppress the square brackets $[\cdot]$ around cohomology classes to simplify the notation.

We adopt the following axiomatic definition for the Stiefel--Whitney classes of a vector bundle $E^k\rightarrow M$. 

\begin{defn}
    Define the full Stiefel--Whitney class $W(E^k) = 1+W_1(E^k)+\cdots$ as the unique element in the cohomology ring $H^*(M,\Z_2)$ satisfying the following 4 properties.
    \begin{enumerate}
        \item If $f:E^k\rightarrow F^k$ is a bundle map then $f^*W(F^k)= W(f^*F^k)$ \hfill (\textbf{Naturality})
        \item $W_p(E^k)\in H^p(B,\Z_2)$ and $W_p(E^k) = 0$ for $p>k$ \hfill (\textbf{Rank})
        \item $W(E^k\oplus F^l) = W(E^k)\smile W(F^l)$ \hfill (\textbf{Whitney Sum Formula})
        \item If $\gamma^1\rightarrow \R P^1$ is the real tautological line bundle then $\pair{W_1}{\mu_{\R P^1}}=1$ \hfill (\textbf{Normalisation})
    \end{enumerate}
\end{defn}

The notation $W_p(E^k)$ in the above indicates that we are referring to the p'th Stiefel--Whitney class of the bundle $E^k\rightarrow M$ specifically. Throughout the rest of this paper this distinction is irrelevant and so the argument is suppressed.

\begin{rmrk}
    It is possible to define $W_k$ in a more constructive way, either as obstruction classes or using the Thom isomorphism. Details on this can be found in either the standard text by Milnor \cite{milnor1974characteristic} or in Bott and Tu \cite{bott2013differential}. 
\end{rmrk}

For the purposes of the eigenbundle identification problem we will specifically need to consider the first Stiefel--Whitney class $W_1$. This is because the Euler class, which we will use to compute the Chern numbers, is only defined for oriented vector bundles. So, given a vector bundle we will we need to be able to check if it is orientable. $W_1$ is understood as the obstruction to orientability as stated in Theorem \ref{thrm:W1Obstruction} \cite{milnor1974characteristic}. 
\begin{thrm}\label{thrm:W1Obstruction}
    A vector bundle $E^k\rightarrow M$ is orientable if and only if its first Stiefel--Whitney class $W_1(E^k)\in H^1(M,\Z_2)$ is zero.
\end{thrm}
Computing $W_1$ will allow us to determine if a bundle is orientable and, if so, prescribe a global orientation for it. This is be discussed in detail in Section \ref{sec:Orientingthebundle}.

We can adopt a similar axiomatic definition of the Euler class for an oriented vector bundle $E^k\rightarrow M$.
\begin{defn}
    The Euler class $e(E^k)\in H^k(M,\Z)$ is the unique cohomology class satisfying the following 4 properties.
    \begin{enumerate}
        \item If $f:E^k\rightarrow F^k$ is an oriented bundle map\footnote{By which, we mean $f$ is orientation preserving on fibers.} then $f^*e(F^k)= e(f^*F^k)$ \hfill (\textbf{Naturality})
        \item If $\bar{E}^k$ refers to the oppositely oriented bundle then $e(\bar{E}^k)=-e(E^k)$ \hfill (\textbf{Orientation})
        \item $e(E^k\oplus F^l) = e(E^k)\smile e(F^l)$ \hfill (\textbf{Whitney Sum Formula})
        \item For $\gamma^1\rightarrow \C P^1$ the tautological line bundle then $\pair{e}{\mu_{\C P^1}}=-1$ \hfill (\textbf{Normalisation})
    \end{enumerate}
\end{defn}
Both the existence of such a unique element, and its geometric interpretation, are not trivial to establish. As above, details can be found in  \cite{milnor1974characteristic,bott2013differential}. 

\begin{rmrk}
    One common characterisation is that $e$ is the obstruction to the existence of a non-vanishing section of the vector bundle. To see this, suppose that $E^k\rightarrow M$ admits a non-vanishing section $s:M\rightarrow E^k$, then we can use this section to construct a trivial sub-bundle of $E^k$ by attaching to each point $p\in M$ the vector space $\text{span}(s(p))\in \pi^{-1}(p)$. Hence $E^k$ splits as a direct sum $E^k = F^{k-1}\oplus \varepsilon$, where $\varepsilon$ refers to the trivial line bunndle over $M$. Therefore the Whitney sum formula tells us that $e(E^k) = e(F^{k-1})\smile e(\varepsilon)$, but since $\varepsilon$ is a trivial bundle $e(\varepsilon)=0$ which implies $e(E^k)=0$ since the cup product of any class with the zero cohomology class gives zero. Hence, if $e(E^k)\neq 0$ then $E^k$ does not admit a non-vanishing section.
\end{rmrk}

Note that these axiomatic definitions of $W_p$ and $e$ are perfectly applicable in the case of discrete vector bundles over a simplicial complex $K$ as well. The only real difference is we think of them classes $W_p\in H^p(K,\Z_2)$ and $e\in H^k(K,\Z)$ in the simplicial cohomology of $K$ rather than the singular cohomology of the base manifold. The algorithms we present in Section $\ref{Section:EulerClass}$ compute these characteristic classes in this simplicial category.

\subsection{First Chern class of a line bundle}\label{sec:FirstChernClass}

Recall that in Section \ref{sec:TopologicalInfo} we proposed to compute $C_1$ for a complex line bundle by evaluating the first Chern class $[c_1]$ on the fundamental class $[\mu]$ of the base space. As such we need to construct $[c_1]$, which we do by appealing to its definition, for the case of line bundles, in terms of the Euler class $[e]$.

Let $L^1_{\C}\rightarrow B$ be a complex line bundle specified by a local trivialisation $\Phi$. In the complex case, each map in the local trivialisation, $\Phi_i:U_i\rightarrow V(1;\C)$ returns a normalised complex vector at each point $x\in U_i$. We will abuse notation and just refer to this vector as $\chi_i(x) \in \C^\infty$. When implementing this theory numerically it is preferable to work with vectors with only real entries\footnote{This is because Scocolla and Perea's approach to construct \v{C}ech cocycles from local trivialisations, which we present below, requires that all matrices be real.}, so we will inject $\gamma:\C^\infty\rightarrow\R^\infty$ by the obvious map $(u_1+iv_1,u_2+iv_2,\ldots)\overset{\gamma}{\mapsto} (u_1,v_1,u_2,v_2,\ldots)$. Therefore, we can treat our complex line bundle as specified by the maps $\gamma\circ \chi_i$ which assign a real line to each point in the base. 

Every complex line bundle is associated to an oriented rank-2 real vector bundle $L^2_\R\rightarrow B$ called its \textit{realisation}. We define the realisation by a local trivialisation as follows:
\begin{defn}
    Let $\Phi_i:U_i\rightarrow V(1;\C)$ with $\Phi_i(x) = [\chi_i(x)]$ be a local trivialisation of a complex line bundle. Then its realisation is defined by a local trivialisation $\Psi_i:U_i\rightarrow V(2,\R)$ with frames $\Psi_i(x) = [\gamma( \chi_i(x)),\gamma( J\chi_i(x))]$, where $J$ refers to the linear complex structure on $\C^\infty$. 
\end{defn}

\begin{rmrk}
    The choice of the ordering $(\chi, J\chi)$ on the basis vectors is important. It defines a canonical local orientation induced on the real fibers by the complex structure on each complex fiber. It can be shown that this local orientation is actually global and so all complex vector bundles are automatically orientable, and canonically oriented \cite{milnor1974characteristic}.
\end{rmrk}

There are many equivalent definitions of the Chern classes of a complex line bundle but the following is particularly useful computationally. 

\begin{defn}
    The first (and only non-trivial) Chern class $[c_1]$ of the line bundle $L^1_\C\rightarrow M$ is the Euler class $[e]$ of its realisation $L^2_\R\rightarrow M$.
\end{defn}

\begin{rmrk}
    This is actually a sufficient definition to construct the Chern classes of a general complex vector bundle via the splitting principle \cite{bott2013differential}. We only require the line bundle case for our current TIMs counting problem 
\end{rmrk}

So, the problem of computing $[c_1]$ for our local polarisation bundle reduces to constructing its realisation and then computing the Euler class. Conveniently, a simple algorithm for computing the Euler class of an oriented rank-2 bundle has been produced by Scocolla and Perea \cite{scoccola}.

Similarly to the discussion above, this construction of the realisation of a complex line bundle and its Chern class $[c_1]$ can be extended into the realm of discrete vector bundles.  A discrete complex line bundle over a simplicial complex $K$ is an assignment of a complex Stiefel matrix $\Phi_{\simp{i}} = [\chi_{\simp{i}}] \in V(1;\C)$ to each vertex $\simp{i}\in K^0$. Its realisation is specified by a discrete local trivialisation assigning $\Psi_{\simp{i}}= [\gamma ( \chi_{\simp{i}}),\gamma( J\chi_{\simp{i}})]\in V(2;\R)$ to $\simp{i}$ and its Chern class is defined as the simplicial Euler class of this associated discrete real bundle. The lifting to the approximate case follows by allowing for the realisation to be a discrete approximate vector bundle.

\section{Calculating the Chern number}\label{Section:EulerClass}

In this section we elaborate on how to calculate the Chern number for our TIMs counting problem. The focus of this section is demonstrating how to translate topological concepts into explicit numerical algorithms.

Assume that we are given a simplicial mesh $K$ surrounding the Weyl point, and for each vertex $\simp{i}\in K^0$ we have obtained a local polarisation vector $\chi_{\simp{i}}\in \C^k\subset \C^\infty$. This constitutes a discrete approximate complex line bundle over $K$. Further, assume that we have constructed a local trivialisation $\Phi$ of the realisation of this bundle\footnote{Which, as explained above, amounts to an array $\{\Phi_{\simp{i}} = [\gamma(\chi_{\simp{i}}),\gamma(J\chi_{\simp{i}})]\}_{\simp{i} \in K^0}$ of $V(2;\R)$ Stiefel matrices indexed by the vertices of our simplicial mesh.}, as was detailed in Section \ref{sec:FirstChernClass}. From $\Phi$ we seek to compute $[c_1] = [e]$, $[\mu_K]$ and the pairing $\text{Ch}_1 = \pair{[c_1]}{[\mu_K]}$. 

\begin{rmrk}
    Throughout this section we will assume that our local trivialisation $\Phi$ defines an approximate discrete vector bundle with small enough error $\epsilon$ that its characteristic classes are consistent with the underlying exact local polarisation bundle. Whether or not this assumption holds depends upon the geometry of the point sample from which $K$ is constructed as well as local gradient of the polarisation bundle. In practice, for a given sample, $\epsilon$ can be computed numerically so that the accuracy of the discrete vector bundle can be checked and a new sample generated if needed.
\end{rmrk} 

Firstly, we will discuss how to generate a \v{C}ech cocycle witnessing the trivialisation, since this will be required to calculate the Euler class. We will then explain how to orient a bundle whose trivialisation is not given to us pre-oriented. Next, we present Scoccola and Perea's algorithm for calculating the Euler class of the bundle, at which point we will have all of the necessary steps to calculate $[c_1]$. We then turn to computing $[\mu_K]$, and specifically how to geometrically orient this class. The section concludes with a brief technical note regarding some specific choices we made when producing our \verb!Julia! implementation of the algorithm from Section \ref{sec:TopologicalInfo}.

\subsection{Constructing the \v{C}ech cocyle}

\begin{figure}
    \centering
    \begin{subfigure}[t]{0.49\textwidth}
    \centering
    \begin{tikzpicture}[line cap=round, line join=round,3d view={105}{10},>=stealth,scale=1.25,transform shape,remember picture]
	\coordinate (x1) at (0,0,0);
	\coordinate (x2) at (0,2,1);
	\coordinate (x4) at (-2,0,2);
    \coordinate (x3) at ($(x2)+(x4)$);
    \draw[draw=black] (x3) circle (0pt) node[above right] {${\Phi_{\simp{j}}}$};
    
	\coordinate (n1) at ($(x2) - (x1)$);
	\coordinate (n2) at ($(x2) - (x3)$); 
	\coordinate (O) at (0,0,0);
	\draw[-latex] (-1,0,0) -- (2.5,0,0) node[pos = 1.05] {$x$};
	\draw[-latex] (0,-1,0) -- (0,2.5,0) node[pos = 1.05] {$y$};
	\draw[-latex] (0,0,-1.6) -- (0,0,3) node[pos = 1.05] {$z$};
	
	\path[draw=black, pattern=vertical lines, pattern color = blue!40, thick, opacity = 0.8] (x1) -- (x2) -- (x3) -- (x4) -- (x1);
    \draw [very thick, blue, ->] (O) -- (x2);
    \draw [very thick, blue, ->] (O) -- (x4);

    \coordinate (z1) at (0,0,0);
	\coordinate (z2) at (1.42,0.314,-1.26);
	\coordinate (z4) at (-0.47,2.53,0.94);
    \coordinate (z3) at ($(z2)+(z4)$);
    \draw[draw=black] (z3) circle (0pt) node[below] {${\Phi_{\simp{i}}}$};

    \tikzset{point/.style = {}}
    \node[point] (A) at (0, 2.6,1.1) {};

	\path[draw=black, pattern=north east lines, pattern color = red!40, opacity = 0.8] (z1) -- (z2) -- (z3) -- (z4) -- (z1);
    \draw [thick, red, ->] (O) -- (z2);
    \draw [thick, red, ->] (O) -- (z4);
    \end{tikzpicture}
    
\end{subfigure}
\begin{subfigure}[t]{0.49\textwidth}
    \centering
    \begin{tikzpicture}[line cap=round, line join=round,3d view={105}{10},>=stealth,scale=1.25,transform shape,remember picture]
	\coordinate (x1) at (0,0,0);
	\coordinate (x2) at (0,2,1);
	\coordinate (x4) at (-2,0,2);
    \coordinate (x3) at ($(x2)+(x4)$);
	\coordinate (n1) at ($(x2) - (x1)$);
	\coordinate (n2) at ($(x2) - (x3)$); 
	\coordinate (O) at (0,0,0);
	\draw[-latex] (-1,0,0) -- (2.5,0,0) node[pos = 1.05] {$x$};
	\draw[-latex] (0,-1,0) -- (0,2.5,0) node[pos = 1.05] {$y$};
	\draw[-latex] (0,0,-1.6) -- (0,0,3) node[pos = 1.05] {$z$};
    
    \draw[draw=black] (x3) circle (0pt) node[above right] {${\Phi_{\simp{j}}}$};
	
	\path[draw=black, pattern=vertical lines, pattern color = blue!40, thick, opacity = 0.8] (x1) -- (x2) -- (x3) -- (x4) -- (x1);
    \draw [very thick, blue, ->] (O) -- (x2);
    \draw [very thick, blue, ->] (O) -- (x4);

    \coordinate (y1) at (0,0,0);
	\coordinate (y2) at (0,2.5,0.5);
	\coordinate (y4) at (-1.5,+0.5,1.5);
    \coordinate (y3) at ($(y2)+(y4)$);
    \draw[draw=black] (y3) circle (0pt) node[above right] {${\Phi_{\simp{i}}\Omega_{\simp{ij}}}$};

    \tikzset{point/.style = {}}
    \node[point] (B) at (0,-0.5,1.1) {}; 

	\path[draw=black, pattern=north east lines, pattern color = red!40, opacity = 0.8] (y1) -- (y2) -- (y3) -- (y4) -- (y1);
    \draw [thick, red, ->] (O) -- (y2);
    \draw [thick, red, ->] (O) -- (y4);
    \end{tikzpicture}
\end{subfigure}
\begin{tikzpicture}[overlay, remember picture]
          \path[->,black,dashed,thick] (A) edge (B);
          \node[above] (b) at ($(A)!0.5!(B)$) {$\Omega_{\simp{ij}}$};
\end{tikzpicture}
    \caption{Diagram showing how the solution to the Procustes problem produces the best possible alignment of the bases for two nearby planes.}
    \label{fig:ProcustesProblem}
\end{figure}
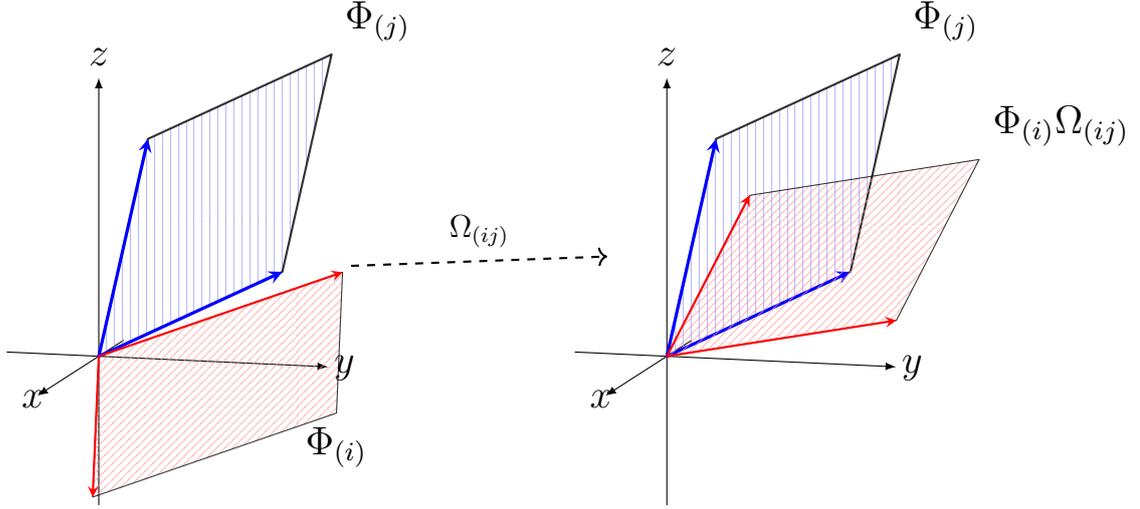

Suppose we are given a simplicial complex $K$ and a local trivialisation $\Phi$ of a approximate discrete vector bundle over it. From the definition of a local trivialisation we know that there exists a $O(k)$ \v{C}ech cocycle $\Omega$ witnessing it. That is, such that
\begin{equation}
    ||\Phi_{\simp{i}}\Omega_{\simp{ij}} - \Phi_{\simp{j}}|| < \epsilon\,,
\end{equation}
for all $\simp{ ij} \in K^1$. Later, we will see that we need the witness $\Omega$ to calculate the Euler class $[e]$. Therefore, we will need to find an explicit set of matrices solving the above inequality. Scocolla and Perea propose a method for this which we reproduce here \cite{scoccola}.

Observe that, if we have an $O(k)$ matrix for which $||\Phi_{\simp{i}}\Omega_{\simp{ij}} - \Phi_{\simp{j}}|| < \epsilon\,$ then any $\Omega'_{\simp{ij}}$ which makes this distance even smaller would also be sufficient to witness $\Phi$. So for each edge in the complex we can just seek the orthogonal matrix which makes this distance as small as possible. Hence we need to solve the following optimisation problem
\begin{equation}\label{eq:Procrustes}
   \Omega_{\simp{ij}} =  \argmin{\Omega \in O(k)}||\Phi_{\simp{i}}\Omega - \Phi_{\simp{j}}||\,,
\end{equation}
which is a well known problem called the Orthogonal Procrustes problem \cite{schonemann1966generalized}. Geometrically, the Procrustes problem answers the question ``given bases for two subspaces, what is the best possible alignment of the two bases?''. Figure \ref{fig:ProcustesProblem} presents a two dimensional example of this idea by indicating that each basis specifies a parallelogram and the choice of optimal rotation maximally aligns the parallelograms while keeping them within the same subspace.

\begin{algorithm}
\setstretch{1.15}
\SetAlgoLined
\DontPrintSemicolon
\KwData{An array of $V(k;\R)$ matrices $\{\Phi_{\simp{i}}\}_{{\simp{i}}\in K^0}$.}
\KwResult{An array of $O(k)$ matrices $\{\Omega_{\simp{ij}}\}_{\simp{ ij }\in K^1}$}
\For{$\simp{ ij}\in K^1$}{
Calculate the SVD, $USV^\dagger = \Phi_{\simp{i}}^\intercal\Phi_{\simp{j}}$\;
$\Omega_{\simp{ij}}\gets UV^\dagger$\;
}
\cprotect\caption{\verb|approxcechcocycle|}\label{alg:approxcechcocycle}
\end{algorithm}

A solution to \eqref{eq:Procrustes} can always be obtained by the Singular Value Decomposition (SVD). Specifically, if $U\Sigma V^\dagger$ is the SVD of $\Phi_{\simp{i}}^\intercal\Phi_{\simp{j}}$ then the solution to the Procrustes problem is $\Omega_{\simp{ij}} = UV^\dagger$. This provides a algorithm for obtaining a \v{C}ech cocycle witnessing $\Phi$. The pseudocode for this is presented as Algorithm \verb!approxcechcocycle!.

\subsection{Orienting the discrete bundle}\label{sec:Orientingthebundle}

Before moving on to discussing calculating $[e]$, we must briefly discuss some data sanitisation. This is because $[e]$ is strictly an invariant of oriented vector bundles but a local trivialisation derived from data we may not be given to us properly oriented. That is, while the frame $\Phi_{\simp{i}}$ attached to each vertex of $K$ will be individually oriented, the chosen orientations may not be consistent across all of the fibers. Further, our data derived vector bundle may not even be orientable and hence we need to be able to determine its orientability to know if computing the Euler class is possible\footnote{As noted above, all complex vector bundles are automatically orientable. and canonically oriented, so this is more relevant for real plane bundles than strictly for complex line bundles. Hence this step is not always necessary but we still consider it a cheap and worthwhile sanitisation process.}. Recall from Theorem \ref{thrm:W1Obstruction} that the orientability of a vector bundle is determined entirely by its first Stiefel--Whitney class $[W_1]$.

For the case of approximate discrete vector bundles over complexes,  we can compute a representative of $[W_1]$ for an approximate $O(k)$ \v{C}ech cocycle $\Omega$  by calculating the determinant of each transition function in the cocycle \cite{scoccola}. Recall that a $1$-cochain, such as $W_1\in C^1(K,\Z_2)$, is a linear function which assigns a $\Z_2$ (binary) value to each edge in $K$. We can therefore represent $W_1$ as a binary array indexed by the edges of the complex $W_1 = \{(W_1)_{\simp{ij}}\}_{\simp{ij}\in K^1}$. In this notation $W_1$ and the determinant of the cocycle matrices $\Omega$ care related as follows
\begin{equation}
    \det \Omega_{\simp{ij}} = (-1)^{(W_1)_{\simp{ij}}}\,.
\end{equation}

This fact suggests the algorithm presented as \verb!sw1!. 

\begin{algorithm}
\setstretch{1.15}
\SetAlgoLined
\DontPrintSemicolon
\KwData{An array of $O(k)$ matrices $\{\Omega_{\simp{ij}}\}_{\simp{ ij }\in K^1}$.}
\KwResult{A binary array $\{(W_1)_{\simp{ij}}\}_{\simp{ij}\in K^1}$. \tcp*{\raggedright Representing $[W_1]\in H^2(K,\Z_2)$}}
\For{$\simp{ ij}\in K^1$}{
    \eIf{$\det(\Omega_{\simp{ij}}) < 0 $}{
        $(W_1)_{\simp{ij}} \gets 1$\;
    }{
        $(W_1)_{\simp{ij}} \gets 0$\;
    }
    
}
\cprotect\caption{\verb|sw1|}\label{alg:sw1}
\end{algorithm}

In principle, the above allows us to check if a given vector bundle is orientable but it does not immediately prescribe how to orient said bundle. Conveniently, $W_1$ does contain the information required to \textit{orient the bundle}. Scocolla and Perea explain how to do this at the \v{C}ech cocycle level, but we prefer to work on the local trivialisations if possible and so we present an alternative construction which is described by Theorem \ref{thrm:orientbundle}.

\begin{thrm}\label{thrm:orientbundle}
    If $\Phi$ is a local trivialisation of a discrete vector bundle over $K$ and $\gamma\in C^0(K,\Z_2)$ is a binary array solving $B_1^\intercal \gamma = W_1$, then $\Phi$ can be transformed into a oriented local trivialisation of the same bundle by switching the orientation of each $\Phi_{\simp{i}}$ for which $\gamma_{\simp{i}} = 1$.
\end{thrm}
\begin{proof}
    Let $\Omega \in \check{C}^1(K,O(k))$ be a set of transition functions which ``witness'' $\Phi$. That is, a discrete \v{C}ech cocycle such that
    \begin{equation}\label{discretewitness}
        \Phi_{\simp{i}} \Omega_{\simp{ij}} = \Phi_{\simp{j}}\,.
    \end{equation}
    Given a \v{C}ech $0$-cochain $\Theta \in \check{C}^0(K,O(k))$ - an assignment of an $O(k)$ matrix to each vertex - we can generate a new set of transition functions describing an equivalent vector bundle by the transformation
    \begin{equation}\label{changeofbasistransform}
        (\Theta\cdot \Omega)_{\simp{ij}} = \Theta_{\simp{i}}\Omega_{\simp{ij}}\Theta_{\simp{j}}^\intercal\,.
    \end{equation}
    ``Orienting the bundle'' amounts to choosing $\Theta$ such that $\Theta\cdot \Omega$ has positive determinant everywhere. Therefore, we need that 
    \begin{equation}\label{OrientatingRequirement}
        1 = \det(\Theta_{\simp{i}})\det(\Omega_{\simp{ij}})\det(\Theta_{{\simp{j}}})\,,
    \end{equation}
    but we recall that $\det(\Omega_{\simp{ij}}) = (-1)^{(W_1)_{\simp{ij}}}$ and $\det(\Theta_{\simp{i}}) = \pm 1$, so we can assume that there exists a simplicial cochain $\gamma \in C^0(K,\Z_2)$ for which $\det \Theta_{\simp{i}} = (-1)^{\gamma_{\simp{i}}}$. Therefore, \eqref{OrientatingRequirement} implies that we need
    \begin{equation}
        1 = (-1)^{\gamma_{\simp{i}}+(W_1)_{\simp{ij}}+\gamma_{\simp{j}}}\,.
    \end{equation}
    Then, equating powers and noting that we are working modulo 2 we have 
    \begin{equation}
        (W_1)_{\simp{ij}} = \gamma_{\simp{i}}+\gamma_{\simp{j}}   \,.
    \end{equation}
    We now recognise the following chain of equalities\footnote{Each of which follows directly from the algebra of simplicial cohomology as discussed in Appendix \ref{AppendixA:TDA}.}
    \begin{equation}
        \gamma_{\simp{i}}+\gamma_{\simp{j}} = \pair{\gamma}{\simp{ i}}+\pair{\gamma}{\simp{ j}} = \pair{\gamma}{\simp{ i}+\simp{ j}} = \pair{\gamma}{B_1 \simp{ ij} } =\pair{B_1^\intercal\gamma}{\simp{ ij}} = (B_1^\intercal\gamma)_{\simp{ij}}\,,
    \end{equation}
    where $B_1$ is the boundary matrix on edges of $K$. This then yields the cochain equation
    \begin{equation}\label{cohomologyrepofSW1}
        W_1 = B_1^\intercal \gamma\,.
    \end{equation}
    
    If the bundle is orientable then Theorem \ref{thrm:W1Obstruction} guarantees that $W_1$ is zero in cohomology and so must be the coboundary of a zero cochain. That is, a solution to \eqref{cohomologyrepofSW1} exists. We can then choose $\Theta$ to be the identity on each $\simp{i}$ such that $\gamma_{\simp{i}}=0$ and a parity flip on the other open sets. This constructs a valid \v{C}ech cochain solving  \eqref{OrientatingRequirement}. 
    
    The final question we need to consider then is ``what local trivialisation $\Phi'$ does $\Theta\cdot \Omega$ witness?''. This just requires a small amount of algebra. Start with \eqref{changeofbasistransform}, which can be rearranged using the orthogonality of the $\Theta$ matrices to 
    \begin{equation}
        \Omega_{\simp{ij}} = \Theta_{\simp{i}}^\intercal (\Theta\cdot\Omega)_{\simp{ij}}\Theta_{\simp{j}}\,.
    \end{equation}
    Substituting this into \eqref{discretewitness} yields 
    \begin{equation}
        \Phi_{\simp{i}} \Theta_{\simp{i}}^\intercal (\Theta\cdot\Omega)_{\simp{ij}}\Theta_{\simp{j}} = \Phi_{\simp{j}} \,, \nonumber\,,
    \end{equation}
    which can be rearranged, again by the orthogonality of $\Theta$, to 
    \begin{equation}
        \Phi_{\simp{i}} \Theta_{\simp{i}}^\intercal (\Theta\cdot\Omega)_{\simp{ij}} = \Phi_{\simp{j}} \Theta_{\simp{j}}^\intercal\,,
    \end{equation}
    which we written in the same form as \eqref{discretewitness}. In other words, our oriented transition functions witness a local trivialisation $\Phi\cdot \Theta^\intercal $ given by acting on each $\Phi_{\simp{i}}$ with the change of basis $\Theta_{\simp{i}}^\intercal$. Since each $\Theta_{\simp{i}}$ is either the identity (where $\gamma_{\simp{i}}=0)$ or a parity flip this transformation amounts to switching the orientation of $\Phi_{\simp{i}}$ wherever $\gamma_{\simp{i}}=1$. Therefore, performing these orientation flips on $\Phi$ will produce $\Phi\cdot \Theta^\intercal$ which we know is an oriented local trivialisation.
\end{proof}

\begin{rmrk}
    In the above we assumed, for simplicity, that $\Phi$ defined an exact discrete vector bundle. The analogous theorem for an approximate discrete vector bundle also holds and adjusting the proof to this case just requires constructing the same argument with the cocycle condition weakened, as long as $\epsilon$ is small enough. Technically, one use of the rotation invariance of the Frobenius matrix norm is also required. 
\end{rmrk}

Theorem \ref{thrm:orientbundle} then allows us to construct an explicit algorithm for taking an unoriented bundle and orienting it. The pseudocode for this algorithm is included as \verb!orientbundle!.

\begin{algorithm}
\setstretch{1.15}
\SetAlgoLined
\DontPrintSemicolon
\KwData{An array of $V(k;\R)$ matrices $\{\Phi_{\simp{i}}\}_{{\simp{i}}\in K^0}$.}
\KwResult{An array of $V(k;\R)$ matrices $\{\Phi'_{\simp{i}}\}_{{\simp{i}}\in K^0}$.  \tcp*{\raggedright Consistently orientated}}
$W_1\gets \verb!sw1!(\verb!approxcechcocyle!(\{\Phi_{\simp{i}} \})$\;
Solve $W_1 = B_0^\intercal \gamma$ for a binary array $\gamma = \{\gamma_{\simp{i}}\}_{\simp{i}\in K^0}$ by Gaussian elimination \;
\tcp*{\raggedright \small{It is possible that no solution exists, in which case the above will fail. This occurs when the bundle is not orientable.}}
\For{$\simp{i} \in K^0$}{
\eIf{$\gamma_{\simp{i}} = 1$}{
$\Phi_{\simp{i}}'\gets \text{Perm}(\Phi_{\simp{i}})$ \tcp*{\raggedright where Perm permutes the first two columns.}
}{
$\Phi_{\simp{i}}'\gets \Phi_{\simp{i}}$\;
}
}

\cprotect\caption{\verb|orientbundle|}\label{alg:orientbundle}

\end{algorithm}

\subsection{Computing the Euler class}

There are many different descriptions of what exactly the Euler class of an oriented vector bundle really is. We described it in Section \ref{Section:Introduction} as the obstruction to constructing a nowhere vanishing section. Alternatively we could describe it as a measure of how \textit{twisted} the bundle is, in the sense that it measures whether we can choose consistent polar coordinates for the fibers on triple intersections of open sets \cite{bott2013differential}. Further, for the specific case of oriented plane bundles, we can interpret it as the obstruction to \textit{lifting} the $SO(2) \simeq S^1$ structure group to $\R$ \cite{scoccola}. Both of the latter characterisations lead to the same procedure for computing the Euler class of plane bundles at the level of \v{C}ech cocycles.






    

    

If $E^2\rightarrow K$ is a real plane bundle which we are given as an exact discrete cocycle $\{\Omega_{\simp{ij}}\}$ over $K$, then, for each $2$-simplex (triangle) $\simp{ijk}\in K^2$ we know that $\Omega_{\simp{ij}}\Omega_{\simp{jk}} = \Omega_{\simp{ik}}$. But the $SO(2)$ matrices $\Omega_{\simp{ij}}$ are naturally isomorphic to the circle $S^1$ by the relation 

\begin{equation}
   \Omega_{\simp{ij}} =  \begin{pmatrix}
       \cos\Theta_{\simp{ij}}& -\sin\Theta_{\simp{ij}}\\\sin\Theta_{\simp{ij}}& \cos\Theta_{\simp{ij}}
   \end{pmatrix}\,,
\end{equation} where the angles $\Theta_{\simp{ij}}$ are specified by a lifting of $S^1\rightarrow \R$, the choice of which is arbitrary. So the cocycle condition requires that $\Theta_{\simp{ij}}+\Theta_{\simp{jk}} = \Theta_{\simp{ik}}+2\pi e_{\simp{ijk}}$ where $e_{\simp{ijk}}\in \Z$. Conventionally, we rewrite this as
\begin{equation}
    e_{\simp{ijk}} = \frac{\Theta_{\simp{ij}}- \Theta_{\simp{ik}} + \Theta_{\simp{jk}}}{2\pi}\,.
\end{equation}

This implies that when we apply the transition functions in a loop, repeatedly changing basis in the fiber, we will return to the representation that we started with after performing a whole number of rotations. $e_{\simp{ijk}}$ just counts the number of rotations we perform. $e_{ijk}$ is clearly an assignment of an integer to each $2$-simplex and so defines a simplicial cochain in $e\in C^2(K,\Z)$. Further, it can be shown that $e$ is a cocycle and actually represents the Euler class of the vector bundle \cite{scoccola,bott2013differential}.


For the case of an $\epsilon$-approximate discrete vector bundle the above discussion almost holds. The only issue is that the cocycle condition is only satisfied up to $\epsilon$ in the Frobenius distance so we cannot just directly equate the angles. Scocolla and Perea showed that as long as the error in the bundle $\epsilon \leq 1$ then $(\Theta_{\simp{ij}}- \Theta_{\simp{ik}} + \Theta_{\simp{jk}})/2\pi$ will be sufficiently close to an integer that we can just round it to the nearest one. Further, they show that when $K$ is the nerve of a cover it reproduces the Euler class of the cover if the cocycle error is sufficiently small. For the exact meaning of ``sufficiently small'' we refer to the original publication \cite{scoccola}. 

\begin{algorithm}
\setstretch{1.15}
\SetAlgoLined
\DontPrintSemicolon
\KwData{An array of $V(k;\R)$ matrices $\{\Phi_{\simp{i}}\}_{\simp{i}\in K^0}$.\tcp*{\raggedright Assumed globally oriented}}
\KwResult{An integer array $\{e_{\simp{ijk}}\}_{\simp{ijk}\in K^2}$. \tcp*{\raggedright Representing $[e]\in H^2(K,\Z)$}}
$\{\Omega_{\simp{ij}}\}_{\simp{ ij }\in K^1}\gets \verb!approxcechcocyle!(\{\Phi_{\simp{i}} \})$\;
\For{$\simp{ij}\in K^1$}{
$\Theta_{\simp{ij}} \gets \arctan\left(\frac{[\Omega_{\simp{ij}}]_{21}}{[\Omega_{\simp{ij}}]_{11}}\right)$ \tcp*{Construct an array of rotation angles}
}
\For{$\simp{ ijk} \in K^2 $}{
    $\Delta\Theta = \text{Round}\left(\frac{\Theta_{\simp{ij}}-\Theta_{\simp{ik}}+\Theta_{\simp{jk}}}{2\pi}\right)$\;
    $e_{\simp{ijk}} \gets \Delta\Theta$\;
}
\cprotect\caption{\verb|eu|}\label{alg:eu}
\end{algorithm}

A pseudocode description of the algorithm for computing the Euler class is provided as Algorithm \verb!eu!. 

\subsection{Computing the Chern numbers}

Given a representative of the first Chern class $c_1$, we recall that the Chern number $\text{Ch}_1$ is defined by evaluated $c_1$ against the fundamental class of the base space $[\mu_K]\in H_2(K,\Z)$. That is
\begin{equation}
    \text{Ch}_1 = \pair{c_1}{\mu_K}\,,
\end{equation}
where we note that the above is in $\Z$ since $c_1\in C^2(K,\Z)$ is a 2-cochain and so pairs with a 2-chain to return an integer. We are implicitly assuming here that $K$ is a simplicial 2-complex. This is not strictly necessary but is the relevant case for our TIM counting problem.

Computing $C_1$ therefore requires obtaining a representative of the fundamental class. Recall that $[\mu_K]$ is defined to be the class generating $H_2(K,\Z) = \Z$, and hence a representative $\mu_K\in C^2(K,\Z)$, written as an integer valued array $\mu_K = \{\mu_{\simp{ijk}}\}_{\simp{ijk}\in K^2}$, will satisfy the equation
\begin{equation}
    B_2\mu_K = 0\,,
\end{equation}
where $B_2$ refers to the $2$-chain boundary matrix. We can solve this linear equation by integer Gaussian elimination, or more accurately reduction to Smith normal form, but there will be two valid solutions\footnote{Technically, there is an entire $\Z$ family of solutions ($k\mu_{K}$ for $k\in \Z$) and a solver could arrive at any one of them. However, it is always possible to reduce to $\pm \mu_{K}$ by repeatedly dividing though by any integer factors which are shared by all of the elements in the array obtained by the solver.} ($\pm \mu_{\text{K}}$) and our solver may obtain either. Note that, these two solutions correspond to the two possible choices of orientation on the surface $K$. 

\subsection*{Orienting $K$}\label{sec:OrientingK}

We call the array $\mu_{\text{top}}=\pm \mu_K$, which our linear solver outputs, the topological fundamental class to distinguish it from $\mu_K$ which will think of as the geometric fundamental class. They are distinct because, as an embedded simplicial complex, $K$ has a geometric orientation defined on it in terms of outward pointing normal vectors. If we want our topological calculation of $\text{Ch}_1$ to agree with the geometric calculation ($\int_M \mathcal{F}$) then we need to ensure that we are using the geometric orientation. $[\mu_{\text{top}}]$ is a already consistent generator of $H_2$, and therefore is \textit{an orientation}, but it may have the wrong global sign. So, we need a procedure to detect this sign and to correct it if is inconsistent. To do this we need to known how the topological and geometric orientations relate to one another.

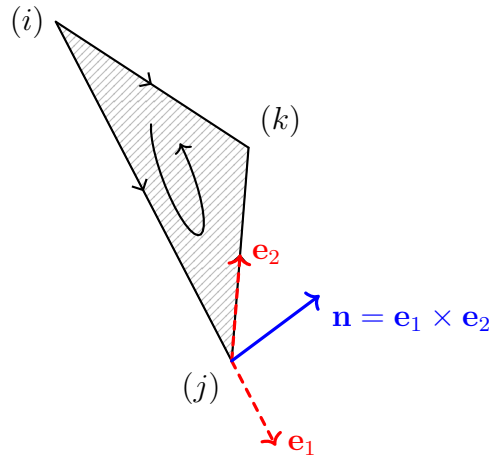
\begin{wrapfigure}{r}{0.4\textwidth}
    \begin{tikzpicture}[line cap=round, line join=round>=stealth,scale=1.2,3d view={140}{30}]

    \coordinate (y1) at (0,0,0);
	\coordinate (y2) at (0,3,-3);
	\coordinate (y3) at (-1.5,1.5,-1.5);

	\path[draw=black, pattern=north east lines, pattern color = black!40, opacity = 0.8] (y1) -- (y2) -- (y3) ;
    \draw [thick, black, ->-] (y1) -- (y2);
    \draw [thick, black, ->-] (y1) -- (y3);
    \draw [thick, black, ->-] (y2) -- (y3);

	\coordinate (n1) at ($(y2) - (y1)$);
	\coordinate (n2) at ($(y3) - (y2)$); 

    \draw [red, very thick, dashed,->] (y2) -- ($(y2)+0.25*(n1)$) node[right] {$\mathbf{e}_1$};
    \draw [red, very thick,dashed,->] (y2) -- ($(y2)+0.5*(n2)$) node[right] {$\mathbf{e}_2$};

    \coordinate (n1crossn2) at (0,1.5,1.5);
    \draw[very thick,blue,->] (y2) -- ($(y2) + (n1crossn2)$) node[below right] {$\mathbf{n} = \mathbf{e}_1\times\mathbf{e}_2$};

    \node[left] (i) at (y1) {$\simp{ i}$};
    \node[below left] (j) at (y2) {$\simp{ j}$};
    \node[above right] (k) at (y3) {$\simp{ k}$};

    \coordinate (c) at ($(y1)+0.5*(n1)+0.5*(n2)$);
    \draw [black,thick,domain=180:450,->] plot ({-0.495-0.15*1.5*sin(\x)}, {1.485+0.15*(3*cos(\x)-1.5*sin(\x))},{-1.485+0.15*(-3*cos(\x)+1.5*sin(\x))});

    \end{tikzpicture}
    \caption{\small{Orientation induced by lexicographical ordering on a 2-simplex.}}
    \label{fig:orientation}
\end{wrapfigure}

Consider a single 2-simplex $\simp{ijk}$ and note that the lexicographical ordering on simplices implies that $i<j<k$. This ordering defines the topological orientation of the simplex. When embedded in $\R^3$ the simplex lies on a plane and a geometric orientation on the simplex amounts to a choice of the direction of the normal vector to the plane. The topological orientation generates a natural choice of normal direction as follows. Let $\mathbf{x}_{\simp{i}},\mathbf{x}_{\simp{j}}$ and, $\mathbf{x}_{\simp{k}}$ be the vertex positions for $\simp{ijk}$. Then, the vectors $\mathbf{e}_1 = \mathbf{x}_{\simp{j}}-\mathbf{x}_{\simp{i}}$ and $\mathbf{e}_2 = \mathbf{x}_{\simp{k}}-\mathbf{x}_{\simp{j}}$ form a basis for the plane of $\simp{ijk}$ and $\mathbf{n} = \mathbf{e}_1\times \mathbf{e}_2$ is a normal vector to the plane. As shown in Figure \ref{fig:orientation} this normal vector is oriented in agreement with the right-hand rule with respect to the $i<j<k$ ordering. The problem then is to determine whether this normal vector is inwards or outwards pointing with respect to the full surface of which $\simp{ijk}$ is only a small section.

This leads us to an interesting computational geometry exercise ``Given a vector $\mathbf{n}\in \R^3$ which is transverse to a compact surface\footnote{For the purposes of computing the $\text{Ch}_1$ of line bundles over spheres, this level of generality is unnecessary. We have included it here highlight that our $\text{Ch}_1$ calculation is possible over any surface. While this is not directly relevant to us, it is relevant as a more general applied topology problem. } $\Sigma\subset \R^3$ and based at a point $\mathbf{x}\in \Sigma$, determine if the vector is inwards or outwards pointing with respect to $\Sigma$''. To answer this, consider the ray $R = \{\mathbf{x}+\mathbf{n} t,t>0\}$ extending outwards from the vector $\mathbf{n}$. It follows that $\mathbf{n}$ is outwards (inwards) pointing if the ray $R$ intersects $\Sigma$ an even (odd) number of times. Therefore we have the that outward pointing vector is given by
\begin{equation}
    \mathbf{n}_{\text{out}} = (-1)^{\# \text{intersections}}\mathbf{n}\,.
\end{equation}

Recall that we are given our surface as a simplicial triangulation $K$. Therefore to count the number of intersections between $K$ and $R$ we just need to check if $R$ intersects each 2-simplex. Note that each 2-simplex $\simp{ijk}$ spans a plane $P_{\simp{ijk}}$ and that $P_{\simp{ijk}}$ and $R$ intersect once and only once\footnote{Assuming that the ray and plane are not parallel which can always be ensured by perturbing the ray.}. Let $\mathbf{w}_{\simp{ijk}}$ be the normal vector to $P_{\simp{ijk}}$ and note that $\mathbf{x}_{\simp{j}}$ is a point on $P_{\simp{ijk}}$. Then the intersection occurs when the displacement between the location along the ray and $\mathbf{x}_{\simp{j}}$ is orthogonal to $\mathbf{w}_{\simp{ijk}}$. That is, where
\begin{equation*}
    \mathbf{w}_{\simp{ijk}}\cdot( \mathbf{x} + t_0\mathbf{n}-\mathbf{x}_{\simp{j}})=0\,,
\end{equation*}
which gives that
\begin{equation}
    t_0 = \frac{\mathbf{w}_{\simp{ijk}}\cdot(\mathbf{x}-\mathbf{x}_{\simp{j}})}{\mathbf{w}_{\simp{ijk}}\cdot\mathbf{n}}\,.
\end{equation}

Finally, we then just need to check if this intersection point $\mathbf{y} = \mathbf{x}+t_0\mathbf{n}$ actually lies within the triangle. This is done by noting that $P_{\simp{ijk}}$ can be thought of as elements of the form $\alpha (\mathbf{x}_{\simp{i}}-\mathbf{x}_{\simp{j}})+\beta (\mathbf{x}_{\simp{k}}-\mathbf{x}_{\simp{j}})$ where the triangle itself is then the subset with $\alpha,\beta\geq0$ and $\alpha+\beta\leq1$. Solving the linear system $\mathbf{y} = \alpha (\mathbf{x}_{\simp{i}}-\mathbf{x}_{\simp{j}})+\beta (\mathbf{x}_{\simp{k}}-\mathbf{x}_{\simp{j}})$ yields $\alpha$ and $\beta$ and hence allows us to determine in the intersection point lies within the triangle. If it does, then we can increment the intersection count otherwise we ignore the $\simp{ijk}$ simplex and move on.

The full procedure described above can be packaged into a single algorithm for constructing the geometric fundamental class from a topological fundamental class. This algorithm has computational complexity $O(\#K^2)$, where $\#K^2$ is the number of 2-simplices, and is presented as \verb!geometricorientation!.

\begin{algorithm}
\setstretch{1.15}
\SetAlgoLined
\DontPrintSemicolon
\KwData{An embedded complex $K\subset \R^3$ and an integer array $\mu_{\text{top}} = \{(\mu_{\text{top}})_{\simp{ijk}}\}_{\simp{ijk}\in K^2}$}
\KwResult{An integer array $\mu_{K} = \{(\mu_{K})_{\simp{ijk}}\}_{\simp{ijk}\in K^2}$}
Search through $\mu_{\text{top}}$ for an $\simp{ijk}$ such that $(\mu_{\text{top}})_{\simp{ijk}}\neq 0$ \;
$\#\text{Intersections}\gets 0$\;
$\mathbf{v} \gets (\mu_{\text{top}})_{\simp{ijk}}\cdot(\mathbf{x}_j-\mathbf{x}_i)\cross (\mathbf{x}_k-\mathbf{x}_j)$ \tcp*{\raggedright Normal vector oriented by $\mu_{\text{top}}$}
\For{$\simp{ lnm} \in K^2\,       \ ( \neq \simp{ ijk }) $}{
    $\mathbf{e}_1 \gets \mathbf{x}_l-\mathbf{x}_n$\;
    $\mathbf{e}_2 \gets \mathbf{x}_m-\mathbf{x}_n$\;
    $\mathbf{w} \gets (-\mathbf{e}_1)\cross \mathbf{e}_2$\;
    $t \gets \frac{\mathbf{w}\cdot(\mathbf{x}_n-\mathbf{x}_j)}{\mathbf{w}\cdot\mathbf{v}}$\;
    \eIf{$t\leq0$ (or undefined)}{
        continue\;
    }{
    $\mathbf{y} \gets t\mathbf{v} + \mathbf{x}_j$\;
    $M\gets [\mathbf{e}_1 \, \   \mathbf{e}_2]$\;
    Solve $\mathbf{y} = M(\alpha, \beta)^T$ for $\alpha,
    \beta$\;
    \If{$\alpha\geq0,\beta\geq0,$ and $\alpha+\beta \leq 1$ }{
        $\#\text{Intersections}\gets \#\text{Intersections} +1$\;
    }
    }
}
$\mu_{K}\gets (-1)^{\#\text{Intersections}}\mu_{\text{top}}$\;
\cprotect\caption{\verb|geometricorientation|}\label{alg:geometricorientation}
\end{algorithm}

\subsection{A minor technical note}\label{subsec:ModpReduction}
This subsection concerns our specific implementation of the algorithms presented above and how they interact with our chosen method of mesh generation. It is largely separate to the rest of this paper and is most relevant to readers interested in numerical applied topology rather than TIMs specifically. Other readers should feel free to skip ahead to our numerical examples in Section \ref{Section:NumericalExperiments}. The one important takeaway is that our implementation calculates $\text{Ch}_1 \Mod{p}$ (for finite prime numbers $p$) rather than the integer $\text{Ch}_1$ itself.

\subsubsection*{Mesh generation by alpha complexes}
In practice, to obtain our triangulation $K$ of a sphere around the Weyl point we adapt tools from Topological Data Analysis (TDA). Specifically, from the sample $X = \{\bm{\lambda_i}\}_{i=1}^N$, of points surrounding the Weyl point, we generate an alpha filtration $A_d(X)$ and then compute its persistent homology $PH(A(X),\Z_p)$ for a chosen finite prime $p$. From the persistent homology we can then read off the range of diameters $d$ over which the alpha complex has the correct homology to be a sphere. We then choose the largest acceptable diameter and take its alpha complex to be our triangulation. This procedure is neither unique nor strictly necessary, but will be sufficient for the purposes of the numerical examples of TIM counting problems presented in Section \ref{Section:NumericalExperiments}. 

There are two advantages to this mesh generation approach. Firstly, it is automatic and is simple to implement in \verb!Julia! because it can be bootstrapped onto the package \verb|Ripserer|, which is the \verb!Julia! language implementation of the Ripser algorithm \cite{Bauer2021Ripser,vcufar2020ripserer}. Secondly, it allows us to include noise in the sample $X$ without altering the TIMs counting calculation. However, there is one major disadvantage, which is that \verb|Ripserer| does not construct the $\Z$ valued boundary matrices $B_0,B_1,B_2,\ldots $ and instead computes their $\Z_p$ reduction\footnote{$\Z_p$ refers to the set of integers taken modulo an integer $p$ and is a field if $p$ is a prime. The $\Z_p$ reduction of a matrix $M$ - or array, in the case of (co)chains - refers to the process of taking each element in the matrix $M_{ij}$ and replacing it with its value modulo $p$, that is with $M_{ij}\Mod{p}$.}, for a chosen prime number $p\geq2$. This is because \verb|Ripserer| is fundamentally a persistent homology package and it turns out that constructing persistent homology requires the coefficients to be taken from a field rather than a ring \cite{ghrist2018homological}.  

The above implies that we cannot solve for $\mu_K$ directly by matrix reduction on $B_2$ but we can obtain its mod $p$ reduction. That is, if $\mu_{K} = \{(\mu_{K})_{\simp{ijk}}\}_{\simp{ijk}\in K^2}$ refers to the integer array representing the true fundamental class, we can obtain the array 
\begin{equation}
    \mu_K \Mod{p} := \{(\mu_{K})_{\simp{ijk}} \Mod{p}\}_{\simp{ijk}\in K^2}\,.
\end{equation}
Performing matrix reduction on the $B_2$ matrix with $\Z_p$ coefficients would produce this array but it can also just be pulled directly from \verb!Ripserer! as the representative cycle of the most persistent $PH_2(A(X),\Z_p)$ class, which is how we obtain our $\mu_K \Mod{p}$.

\subsubsection*{Reduction mod $\Z_p$}

If we were to obtain explicit representations of both $\mu_{K}$ and $c_1$ in the basis of simplicial 2-chains and 2-cochains respectively, then we could calculate $\text{Ch}_1$ as
\begin{equation}\label{C1calc}
    \text{Ch}_1 = \pair{c_1}{\mu_K}\,.
\end{equation}
However, as explained above, in our implementation we don't have access to the integer fundamental class $\mu_K$ but only $\mu_K \Mod{p}$. Hence, we cannot strictly compute $\text{Ch}_1$ using this method\footnote{Since calculating a pairing between a chain and cochain is effectively a dot product between vectors, and dot products only make sense if the two vectors use the same coefficients.}, but we can compute $\text{Ch}_1 \mod{p}$ for every finite prime $p$. All we need to do is take the mod p reduction of our $c_1$ class and then calculate 
\begin{equation}
    \text{Ch}_1 \Mod{p} = \pair{c_1 \Mod{p}}{\mu_K \Mod{p}} = (c_1 \Mod{p})^\intercal(\mu_K \Mod{p})\,.
\end{equation}
Calculating $c_1 \Mod{p}$ is straightforward since the algorithm \verb!eu! outputs the Chern class as an integer array $c_1 = \{(c_1)_{\simp{ijk}}\}_{\simp{ijk}\in K^2}$ which we can just reduce element-wise mod $p$ to obtain $c_1 \Mod{p} = \{(c_1)_{\simp{ijk}} \Mod{p}\}_{\simp{ijk}\in K^2}$.

Using this method, exactly determining $\text{Ch}_1$ would require redoing this calculation for every prime but it is usually sufficient to just perform the calculation for a few small primes. Certainly, for the purposes of showing a bundle is non-trivial it is sufficient to just show that $\text{Ch}_1$ is not zero mod a single prime. Otherwise, we can consider the set of primes in order and calculate $\text{Ch}_1\Mod{p}$ for each of them. When the prime is larger than $\text{Ch}_1$ this calculation will stabilise at the correct value. Taking the standard lift of $\Z_p$ to $\Z$ gives our result for $\text{Ch}_1$.

\section{Numerical examples}\label{Section:NumericalExperiments}

To demonstrate the algorithms discussed above we will apply them to a few standard examples from the theory of vector bundles and topological waves. Firstly, we will check that algorithm for computing $C_1$ correctly calculates this invariant for the tautological line bundle over the complex projective space. Then we will consider two examples of topological edge modes and show that the $C_1$ calculation can be performed numerically and correctly predicts the number of edge modes. 

Let us demonstrate the application of the above algorithms to three examples: one mathematical example of an abstract Chern number calculation and two examples of TIM counting in physical problems. In each example, we  apply the full algorithm described in Section \ref{sec:TopologicalInfo}. 

\subsection{The tautological line bundle over the Riemann Sphere}

As a check on the validity of our calculation of $\text{Ch}_1$ for a complex line bundle we can confirm that the normalisation axiom of the Euler class is satisfied, which implies that $\text{Ch}_1=-1$ for the tautological complex line bundle. Specifically, consider $\C P^1$ which is defined as the set of complex lines in $\C^2$. The tautological line bundle $\gamma^1\rightarrow \C P^1$ is the line bundle defined by assigning each complex line in $\C P^1$ to itself. We can make an explicit representation of this line bundle using the Riemann sphere. 

Choose any point in $(z_1,z_2)\in \C^2$, and note that this point resides on the complex line of points of the form $\lambda(z_1,z_2)$ with $\lambda\in \C$. As long as $z_2\neq0 $, we can choose $\lambda = z_2^{-1}$ and hence find a representative of the same complex line of the form $(z_1,1)\in \C^2$. This continuously maps $\C P^1$ onto $\C\cup \{\infty\}$  if we also send the complex line $(z_1,0)$ sent to the point at infinity. We can then send $\C\cup \{\infty\}$ to the Riemann sphere by stereographic projection.  

Now, we can define an explicit representation for the tautological line bundle over the Riemann sphere. Let $(\theta,\phi)$ be the usual spherical coordinates, then performing an orientation preserving stereographic projection to $\C\cup\{\infty\}$ before mapping this point to the associated normalised $\C^2$ vector gives a map $\chi:S^2\rightarrow V(1;\C)$,

\begin{equation}\label{TautLineFiber}
    \chi(\theta,\phi) = \frac{1}{\sqrt{4\tan^2\left(\frac{\pi-\theta}{2}\right)+1}}\begin{pmatrix}
        2\tan\left(\frac{\pi-\theta}{2}\right)e^{-i\phi}\\1
    \end{pmatrix}\,.
\end{equation}
At $\theta=0$ this map is singular since it approaches
\begin{equation}
    \chi \longrightarrow \begin{pmatrix}
        e^{-i\phi}\\0
    \end{pmatrix}\,,
\end{equation}
the phase of which varies with $\phi$. This is an example of a phase singularity in a complex vector bundle. For our purposes here the singularity is not an issue as we can choose our sample of points to not include the north pole.

To compute $\text{Ch}_1$ for this bundle we follow the approach described in Section \ref{sec:TopologicalInfo} and apply the algorithms discussed in Section \ref{Section:EulerClass} in turn.

That is, first we choose a random sample of 200 points\footnote{The specific choice of ``200'' here is arbitrary. Any number of points will suffice as long as the error in the associated approximate discrete vector bundle is small enough. We have observed, heuristically, that $\gtrapprox 50$ points is usually sufficient.} $\{\bm{\lambda}_i\in \R^3\}_{i= 1}^{200}$ from the sphere and use them as the vertices to generate a simplicial mesh $K$. Then, we compute the polar coordinates $\{(\theta_i,\phi_i)\}_{i=1}^{200}$ of the vertices so that we to construct a complex line bundle by assigning the $\C^2$ vector given by \eqref{TautLineFiber} to each point. Specifically, we assign the Stiefel matrix $\Phi_{\simp{i}} = [\gamma(\chi(\theta_i,\phi_i)),\gamma(J\chi(\theta_i,\phi_i))]$ to each vertex $\simp{i}\in K^0$. This assignment defines a local trivialisation of the approximate discrete vector bundle over $K$ associated to $\gamma^1\rightarrow \C P^1$. 

Next, we pass the array of Stiefel matrices $\Phi$ to \verb!approxcechcocycle! to compute an array $\Omega = \{\Omega_{\simp{ij}}\}_{\simp{ij}\in K^1}$ of $SO(2)$ matrices representing a discrete \v{C}ech cocycle for this bundle. We then pass $\Omega$ to \verb!eu! to obtain an integer array $c_1 = \{(c_1)_{\simp{ijk}}\}_{\simp{ijk}\in K^2}$ representing the first Chern class of $\gamma^1\rightarrow \C P^1$. 

We then choose\footnote{The need for this choice is discussed in Section \ref{subsec:ModpReduction}.} a finite prime $p$ and solve for the topological fundamental class $\mu_{\text{top}} \Mod{p}$ of our mesh $K$ as a $\Z_p$ valued array. This array can then be passed \verb!geometricorientation! to produce the geometrically oriented fundamental class $\mu_K \Mod{p}$ as a $\Z_p$ array.

Finally, we calculate Chern number mod $p$. For the case of $p=3$ this gives
\begin{equation}
    \text{Ch}_1 \Mod{3} = \pair{c_1\Mod{3}}{\mu_{K}\Mod{3}}  = 2 \Mod{3}\,.
\end{equation}
But $2=-1 \Mod{3}$ and so this does not contradict the normalisation axiom of the Euler class. Repeating this calculation for several different finite primes $p= 3,5,7,\ldots$ we repeatedly obtain $\text{Ch}_1 = -1 \Mod{p}$ and so confirm that $\text{Ch}_1=-1$.

\subsection{TIM counting examples}\label{subsec:TIMsCountingExamples}

We will now verify two spectral flow calculations from the literature on TIMs in fluids and plasmas. Specifically we will consider the Shallow Water (SW) equations and the Topological Langmuir Cyclotron Wave (TLCW). Both of these calculations proceed in a similar fashion. We have a state of $k$ fields $\ket{\psi}$ and a Hamiltonian $\hat{H}$ modelling waves near an interface in an effectively 2D system. We use $x$ and $y$ to refer to the direction across the interface and along the interface respectively. We assume the system is homogenous in $y$ and so can consider waves propagating with wavevector $k_y$ along the interface independently. Treating $k_y$ as a spectral flow parameter in $\hat{H}$ we seek to calculate the net spectral flow $\Delta \mathcal{N}_n$ into the $n$'th band. Recalling the generalised SFM correspondence \eqref{GeneralisedSFM} we have 
\begin{equation}
    \Delta \mathcal{N}_n = -\text{Ch}_1^n\,,
\end{equation}
where $\text{Ch}_1^n$ is the Chern number of the local polarisation bundle of the Weyl symbol $H$, associated to the $n$'th band, taken over a sphere around a degeneracy point\footnote{Here we are assuming that there is only one Weyl point in the phase space. This will not always be the case but is true for the examples we consider and is a helpful simplification.}.

The approach to compute the spectral flow for the SWE and TLCW problems is effectively the same but with differing Hamiltonians and hence different Weyl symbols. The procedure is as follows:
\begin{enumerate}
    \item From the Hamiltonian determine the Weyl symbol $H(x,k_x,k_y)$ using the operator-symbol correspondence..
    \item Determine the location of the isolated degeneracy point $(\tilde{x},\tilde{k}_x,\tilde{k}_y)$. That is, the point where two eigenvalues of $H$ collide.
    \item Choose a sample of $N\gtrapprox 50$ points\footnote{As discussed above, the chosen number of points in arbitrary but we have found that between $\approx50$ and $\approx 200$ is usually sufficient.} $X = \{(x^i,k_x^i,k_y^i)\}_{i=1}^N$ from a uniform distribution over a small sphere surrounding the degeneracy point. 
    \item Generate a simplicial mesh $K$ using the points of $X$ as the vertices of the mesh.
    
    \item At each $(x^i,k_x^i,k_y^i)\in X$ compute the eigenvalues and eigenvectors of the Weyl symbol $H(x^i,k_x^i,k_y^i)$ and sort them by increasing eigenvalue.
    
    Call the normalised eigenvector associated to the $n$'th eigenvalue $\chi_i^n$. 
    
    Note that we perform all of the following steps for each $n$ since each branch in $H$'s spectrum defines a different local polarisation bundle.
    \item  Generate from these vectors an approximate discrete local trivialisation $\Phi^n$ of the associated rank 2 real vector bundle. 
    
    That is, construct the array of $V(2;\R)$ Stiefel matrices $\Phi^n = \{ \gamma(\chi_i^n),\gamma(J\chi_i^n)\}_{i \in 1}^N$
    
    \item Compute the discrete \v{C}ech cocycle of this bundle $\Omega^n = \verb!approxcechcocycle!(\Phi^n)$, which is an array $\Omega^n = \{\Omega_{\simp{ij}}^n\}_{\simp{ij}\in K^1}$ of $SO(2)$ matrices indexed by the edges of $K$.
    \item Compute the integer array $e^n = \{e^n_{\simp{ijk}}\}_{\simp{ijk}\in K^2} = \verb!eu!(\Omega^n)$ representing the Euler class, which is also the Chern class.
    \item Choose a prime number $p$ and solve for a null vector $\mu_{\text{top}} \Mod{p}$ of the boundary matrix $B_2 \Mod{p}$.
    \item Compute the geometric fundamental class $\mu \Mod{p} = \verb!geometricorientation!(\mu_{\text{top}}\Mod{p})$, stored  as a $\Z_p$ array $\mu = \{\mu_{\simp{ijk}}\Mod{p}\}_{\simp{ijk}\in K^2}$
    \item Calculate $\text{Ch}_1^n \Mod{p} = \sum_{\simp{ijk}\in K^2}(\mu_{\simp{ijk}}\Mod{p})\cdot (e^n_{\simp{ijk}}\Mod{p})$.
    \item Repeat $9\rightarrow 11$ for several different primes.
\end{enumerate}

\subsubsection{The shallow water equations }

We deal first with the case of the SW equations which is a local model of the flow of water on the surface of a rotating sphere near its equator\footnote{The global case is more complicated and has been investigated by Perez \cite{perez2025topology}.} \cite{venaille2023ray,delplace2017topological}. The Hamiltonian for the SW equations, assuming plane waves along the equator, is 
\begin{equation}
    \hat{H} = \begin{pmatrix}
        0 & -if(\hat{x}) & \hat{k}_x\\
        if(\hat{x}) & 0 & k_y\\
         \hat{k}_x & k_y & 0
    \end{pmatrix}\,,
\end{equation}
where: $\hat{x}$ is the position operator, $\hat{k}_x = -i\partial_x$ is the wavenumber operator, $f$ is a function which specifies our model of the Coriolis force, and $k_y$ is a spectral flow parameter corresponding to the wavevector along the interface. This model is known to have three bands: positive Poincare modes, negative Poincare modes, and Rossby waves. It also exhibits a spectral flow of $\Delta \mathcal{N} = +2$ into the positive Poincare band. These spectrally flowing dispersion curves correspond to modes uni-directionally propagating equatorial waves called the Kelvin and Yanai waves \cite{delplace2017topological}. 

We will assume that $f(0)=0$ and that $f$ is monotonically increasing so that we can change coordinates and use $(f,k_x)$ as our phase space coordinates\footnote{We could, equivalently, Taylor-expand $f$ near $0$, keeping only the linear term and then rescale the coordinates.}. We can then read off the Weyl symbol as 
\begin{equation}
    H(f,k_x,k_y) = \begin{pmatrix}
        0 & -if & k_x\\
        if & 0 & k_y\\
         k_x & k_y & 0
    \end{pmatrix}\,.
\end{equation}

The Weyl symbol has three eigenvalues $\omega_{\pm}(f,k_x,k_y) = \pm\sqrt{f^2+k_x^2+k_y^2}$ and $\omega_0 = 0$ which correspond to the positive and negative Poincare modes, and the Rossby waves respectively. These define three eigenbundles which we call $E_\pm$ and $E_0$. There is a single degeneracy at $(0,0,0)$, where all three eigenvalues coincide. The SFM correspondence tells us that the net spectral flow into a band will be equal to the Chern number associated to a sphere surrounding this point. Following the procedure discussed above we computed the Chern number of this point for each branch in the dispersion relation modulo several different primes. The results of this calculation are presented in Table \ref{tab:SWEquationsChernNumbers}. We note that the Chern number of the positive Poincare $+$ band was found to be $\text{Ch}^+_1=-2$. This agrees with the result of analytic calculations in this model \cite{delplace2022berry} and with the observed $+2$ spectral flow into this band.  

\begin{table}[t]
    \centering
    \begin{tabular}{c|c|c|c|c}
        Branch & $\text{Ch}_1 \Mod{3}$ & $\text{Ch}_1 \Mod{5}$ & $\text{Ch}_1 \Mod{7}$ & $\text{Ch}_1$ \\\hline\hline
        $-$ & 2 & 2 & 2 & 2\\\hline
        $0$ & 0 & 0 & 0 & 0\\\hline
        $+$ & 1 & 3 & 5 & -2
    \end{tabular}
    \caption{Results of computing the Chern numbers of the Weyl point of the SW equations by the algorithm described in Section \ref{subsec:TIMsCountingExamples}.}
    \label{tab:SWEquationsChernNumbers}
\end{table}

\subsubsection{The topological Langmuir cyclotron wave}

The SW equations are a particularly simple case in the sense that the Weyl symbol of the Hamiltonian can be written as an element of an $\mathfrak{su}(2)$ subalgebra \cite{delplace2022berry}. To test our algorithm on a more complicated problem we now investigate an example from plasma physics, the Topological Langmuir Cyclotron Wave (TLCW). This is a TIM which emerges at a density boundary in an under-dense magnetised plasma and is created by the Weyl point where the Langmuir and cyclotron branches intersect. The TLCW was investigated by Fu and Qin \cite{qin2023topological} who observed a spectral flow between the Langmuir and cyclotron bands numerically and analytically calculated the Chern number using a local model near the Weyl point. They found a spectral flow of $\Delta \mathcal{N} = -1$ for a density wall with positive slope\footnote{Fu and Qin \cite{qin2023topological} choose a negative density slope in their original work. This flips the signs in all of their spectral flow and Chern number calculations and so what they call a +1 spectral flow is a -1 spectral flow in the more standard convention we are adopting here.} and so we expect a Chern number of $C_1 = +1$ for the higher frequency of the two bands. We can now confirm this using our numerical method.

The TLCW is an edge mode in a linearised cold magnetised plasma with stationary ions. There are three relevant vector fields: $\mathbf{v},\mathbf{E},$ and $\mathbf{B}$ representing the electron fluid velocity, electric field, and magnetic field respectively. The full state $\ket{\psi}$ is then comprised of the nine component fields. We assume that there is a constant background magnetic field in the $z$ direction. The Weyl symbol, after nondimensionalisation, for this model is written in block matrix form as 
\begin{equation}
    H(\mathbf{x},\mathbf{k}) = \begin{pmatrix}
        i \bm{e}_z\times & -i w_p(\mathbf{x}) & 0\\
        iw_p(\mathbf{x}) & 0 & -\mathbf{k}\times\\
        0 & \mathbf{k}\times & 0
    \end{pmatrix}\,,
\end{equation}
where: $\bm{e}_z$ refers to the unit vector in the $z$ direction, scalars are implied to be multiplied by $\Id_3$, and the $\times$ symbol refers to antisymmetric matrix which implements the cross product:
\begin{equation}
    \mathbf{b}\times = \begin{pmatrix}
        0 & -b_z & b_y\\
        b_z & 0 & -b_x\\
        -b_y & b_x & 0
    \end{pmatrix}\,.
\end{equation}
 Here $w_p(\mathbf{x}) = \sqrt{\frac{4\pi n(\mathbf{x}) e^2}{m_e\Omega^2}}$ refers to the plasma frequency normalised by the electron cyclotron frequency. We focus on the case of a one inhomogenous dimension $w_p(\mathbf{x}) = w_p(x)$ and therefore can Fourier-transform along both $y$ and $z$, the homogenous directions. This leaves an effectively four dimensional parameter space $(x,k_x,k_y,k_z)$ and so the SFM correspondence does not immediately apply. However, we can fix $k_z$ and then will be left with only a three dimensional parameter space $(x,k_x,k_y)$ in which to find our Weyl point. For a given $k_z$, the Weyl point occurs where $(k_x,k_y)= 0$ and at the frequency
\begin{equation}
    w_p(x) = w_{pc} = \frac{\sqrt{k_z^4+4k_z^2}-k_z^2}{2}\,.
\end{equation}

As before, we will assume that near to the Weyl point the $\omega_p$ profile is injective so we can use $\omega_p$ instead of $x$ as the coordinate of our space. Constructing a radius 0.1 sphere around $(\omega_p,k_x,k_y) =(\omega_{pc},0,0)$ in phase space we have a set of nine eigenbundles $E_i$ for the Weyl symbol $H$: a central $\omega=0$ mode $E_0$ which we are not concerned with; four positive frequency modes $E_1,\ldots E_4$; and four symmetric negative frequency modes $E_{-4}\ldots E_{-1}$. Applying the algorithm explained in Section \ref{subsec:TIMsCountingExamples} to calculate $\text{Ch}^n_1$ for each of these branches we obtained the results presented in Table \ref{tab:TLCWEquationsChernNumbers}. Note that between the $E_1$ and $E_2$ bands we find Chern numbers of $-1$ and $+1$ respectively, indicating the presence of a net $-1$ spectral flow from band $E_1$ into $E_2$ as $k_y$ is varied. This topological edge mode is the TLCW, and the result of $\text{Ch}^1_1 = -1$ agrees with the analytic results of Fu and Qin \cite{qin2023topological}. 

\begin{table}[t]
    \centering
    \begin{tabular}{c|c|c|c|c}
        Branch & $\text{Ch}_1 \Mod{3}$ & $\text{Ch}_1 \Mod{5}$ & $\text{Ch}_1 \Mod{7}$ & $\text{Ch}_1$ \\\hline\hline
        $-2$ & 2 & 4 & 6 & -1\\\hline
        $-1$ & 1 & 1 & 1 & 1\\\hline
        $0$ & 0 & 0 & 0 & 0\\\hline
        $1$ & 2 & 4 & 6 & -1\\\hline
        $2$ & 1 & 1 & 1 & 1
    \end{tabular}
    \caption{Results of computing the Chern numbers of the Weyl point of the TLCW the algorithm described in Section \ref{subsec:TIMsCountingExamples}. The $E_{-4},E_{-3},E_{3},E_{4}$ branches are all ignored here because all of their Chern numbers were zero.}
    \label{tab:TLCWEquationsChernNumbers}
\end{table}

\section{Extraction of Chern numbers from data}\label{Section:DataAnalysis}

Since the algorithm for computing $\text{Ch}_1$ which was presented above requires only a finite sample of the eigenvectors surrounding the Weyl point, it is also applicable to experimental data. Consider the case of a TIMs problem, which as discussed above, amounts to solving a vector PDE of the form 

\begin{equation}\label{Schrodinger}
    \omega \ket{\Psi(\omega,k_y)} = \hat{H}(\hat{x},\hat{k}_x,k_y)\ket{\Psi(\omega,k_y)}\,.
\end{equation}
Suppose that we have experimentally measured $\ket{\Psi(\omega,k_y)}$ and aim to extract $\text{Ch}_1$ from this data for a given Weyl point\footnote{We can interpret this as a test on the validity of our model Hamiltonian. If the measured $\text{Ch}_1$ disagrees with the calculation from the $H(x,k_x,k_y)$ symbol then that is evidence that our model is incomplete.}. 

To perform such an extraction, we will make some assumptions upon the experimental regime. Firstly, the measured data is assumed to describe a quasi-scalar wave. That is, there exists some vector pseudodifferential operator $\hat{\chi}$ such that $\ket{\Psi} = \hat{\chi}\ket{\psi}$, where $\ket{\psi}$ represents a scalar field. We also assume that the resolution is high enough, and the size of the measured domain large enough relative to the variation scale of $\hat{H}(\hat{x},\hat{k}_x,k_y)$ that we can perform our analysis in the geometrical gptics (GO) limit\footnote{The exact requirements for this to hold will depend upon the experimental setup in question.}. Under this construction, it can be shown that, at leading order in GO, \eqref{Schrodinger} requires that the symbol $\chi(x,k_x,k_y)$ of $\hat{\chi}$ be an eigenvector of $H(x,k_x,k_y)$ with eigenvalue $\omega$ (for details on this construction see the paper by Venaille \cite{venaille2023ray}). So if we can extract the operator $\hat{\chi}$ from the solution $\ket{\Psi(\omega,k_y)}$ we obtain the eigenvectors of $H(x,k_x,k_y)$ at each point in phase space and can then assemble a discrete vector bundle.

\subsection{General formulation}\label{Datageneral}

The core question is how to extract $\hat{\chi}$ from the measured fields. In principle, this can be done using the Wigner matrix. Assuming that our fields are quasi-scalar waves, their density matrix can be represent as 
\begin{equation}
    \hat{\rho} = \ket{\Psi}\bra{\Psi} = \hat{\chi}\ket{\psi}\bra{\psi}\hat{\chi}^\dagger\,,
\end{equation}
The Weyl symbol of this operator is a matrix valued function called the Wigner matrix $W$ of the state $\ket{\Psi}$ \cite{folland2016harmonic}. At leading order in GO we have that
\begin{equation}
    W(x,k_x,k_y) = \chi W_\psi\chi^\dagger\,,
\end{equation}
where the terms on the right refer to the Weyl symbols of $\hat{\chi}$ and $\ket{\psi}\bra{\psi}$, which we think of as a scalar Wigner function $W_\psi$. Since $W_\psi$ is scalar we can rearrange the above as 
\begin{equation}\label{AsymptoticWigner}
    W(x,k_x,k_y) = W_\psi\chi\chi^\dagger\,,
\end{equation}
from which we conclude that at each point in phase space the Wigner matrix is a projector onto the local eigenvector of $H$. 

So, suppose we are given an experimentally measured $\ket{\Psi}$, under the same assumptions as above, we can compute its Wigner matrix by taking the full Wigner transform
\begin{equation}
    W(x,k_x,k_y) = \frac{1}{2\pi}\int \Psi(x+s/2,k_y)\Psi^\dagger(x-s/2,k_y)e^{-ik_xs}\, ds\,.
\end{equation}
Due to experimental noise - and the fact that \eqref{AsymptoticWigner} is only an asymptotic equality - we do not expect that this $W$ will satisfy \eqref{AsymptoticWigner} exactly but instead can assume that there exists normalised vectors $\chi$'s such that
\begin{equation}
    W \approx W_\psi \chi\chi^\dagger\,.
\end{equation}
Therefore, our experimentally measured Wigner function should be ``almost a projector''. By seeking the closest projector to $W$ at each point in phase space we can obtain the eigenvectors $\chi$ from which the vector bundle needed to calculate $\text{Ch}_1$ is assembled. This implies that $\chi$ is found as
\begin{equation}\label{minimisation}
    \chi(x,k_x,k_y) = \argmin{\xi}||W(x,k_x,k_y) - \xi \xi^\dagger)||\,,
\end{equation}
where the norm on matrices is taken to be the $L_2$ operator norm and we implicitly normalise the solution vector. Since the Wigner matrix is, by construction, Hermitian one can search for $W$ in the form $W = \sum_i\lambda_ie_ie_i^\dagger$, where $\lambda_i$ and $e_i$ are the eigenvalues and orthonormalised eigenvectors of $W$ respectively.  If we order the eigenvectors by decreasing eigenvalue then the solution to \eqref{minimisation} is  
\begin{equation}\label{chiestimate}
    \chi = e_1\,.
\end{equation}
In other words $\chi$ is the normalised eigenvector of $W$ corresponding to the largest eigenvalue.

The above construction includes a natural measurement of its error which we can use to gauge the validity of our quasi-scalar assumption. At each point in phase space we can compute 
\begin{equation}\label{errorinW}
    \sigma(x,k_x,k_y) = ||\tilde{W}^2 - \tilde{W}||\,,
\end{equation}
where $\tilde{W}=W/||W||$. To the extent that \eqref{AsymptoticWigner} holds, $\tilde{W}$ will be a projector and therefore $\tilde{W}^2=\tilde{W}$, so $\sigma$ will be zero. Therefore, letting the GO parameter be $\varepsilon$ we can say that our approximation is accurate, and hence our of $\chi$ is reliable only where 
\begin{equation}
    \sigma(x,k_x,k_y) \ll \varepsilon\,.
\end{equation}

\subsection{Example: 2D Dirac equation}

To test the above method, let us consider fields governed by the 2D Dirac equation
\begin{equation}\label{DiracEquation}
    i\partial_t \begin{pmatrix}
        \psi_1\\\psi_2
    \end{pmatrix}= \begin{pmatrix}
        x & -i\partial_x -ik_y\\
        -i\partial_x+ik_y & -x
    \end{pmatrix}\begin{pmatrix}
        \psi_1\\\psi_2
    \end{pmatrix}\,,
\end{equation}
which describes two complex fields $\psi_1,\psi_2$ which vary over 1D spacetime $(x,t)$ with a spectral-flow parameter $k_y$, representing the component of the wavenumber directed along the boundary. 

\begin{rmrk}
    The spectrum of \eqref{DiracEquation} was presented in Section \ref{Section:Introduction}, where we observed that there is an $\omega>0$ band, an $\omega <0$ band and a TIM generating a net $\Delta\mathcal{N} = 1$ spectral flow into the positive band. This implies, from the SFM correspondence \eqref{SFMCorrespondence}, that the Chern number of the local polarisation bundle associated to the $\omega>0$ band should be $\text{Ch}_1 =  -1$.
\end{rmrk}

This model is advantageous because \eqref{DiracEquation} can be solved exactly in the Hermite basis \cite{delplace2022berry}. The solution is as follows. Let us define two new fields $\phi_1,\phi_2$ by 
\begin{equation}
    \begin{pmatrix}
        \phi_1\\\phi_2
    \end{pmatrix} = e^{-i\frac{\pi}{4}\sigma_x}\begin{pmatrix}
        \psi_1\\\psi_2
    \end{pmatrix}\,,
\end{equation}
where $\sigma_x$ is the first Pauli matrix. This set of rotated fields satisfies the PDE 
\begin{equation}\label{DiracRotated}
        i\partial_t \begin{pmatrix}
        \phi_1\\\phi_2
    \end{pmatrix}= \begin{pmatrix}
        k_y & i\hat{a}^\dagger\\
        -i\hat{a} & -k_y
    \end{pmatrix}\begin{pmatrix}
        \phi_1\\\phi_2
    \end{pmatrix}\,,
\end{equation}
where we have introduced the raising and lowering operators $\hat{a} = \hat{x}+i\hat{k}$ and  $\hat{a}^\dagger = \hat{x}-i\hat{k}$.

Let us introduce the basis of standard Hermite functions
\begin{equation}
    f_n(x) = \frac{1}{\sqrt{2^n n!}}\left(\frac{1}{\pi}\right)^{1/4}e^{-x^2/2}H_n(x)\,,
\end{equation}
which satisfy 
\begin{align}
    \hat{a}f_n &= \sqrt{2n}f_{n-1}\,,\\
    \hat{a}^\dagger f_n &= \sqrt{2(n+1)}f_{n+1}\,.
\end{align}
Using these one can write the general solution to \eqref{DiracRotated} as
\begin{align}
    \begin{pmatrix}
        \phi_1\\\phi_2
    \end{pmatrix} &= \gamma \begin{pmatrix}
        f_0(x)\\0
    \end{pmatrix}e^{-ik_y t}\\ &+\sum_{n=1}^\infty \frac{1}{\sqrt{2\omega_n}}\left(\alpha_n\begin{pmatrix}
        \sqrt{\frac{2n}{\omega_n-k_y}}f_n(x)\\-i\sqrt{\omega_n-k_y}f_{n-1}(x)
    \end{pmatrix}e^{-i\omega_n t}+\beta_n\begin{pmatrix}
        \sqrt{\frac{2n}{\omega_n+k_y}}f_n(x)\\i\sqrt{\omega_n+k_y}f_{n-1}(x)
    \end{pmatrix}e^{i\omega_n t}\right)\,,\nonumber
\end{align}
where $\omega_n = \sqrt{k_y^2+2n}$, and $\gamma,\alpha_n,$ and $\beta_n$ are parameters which depend on the initial conditions. The first term, which corresponds to $\omega=k_y$, is the TIM. 

Assuming the field is monochromatic in time with frequency $\omega>\sqrt{2}$, its spatial spectrum in $y$ consists of discrete wavenumbers $k_y = k_y^{(n)}$ where
\begin{equation}\label{Quantisedky}
    k_y^{(n)}:= \pm \sqrt{\omega^2-2n}\,,
\end{equation}
and $n$ is restricted to integers from $1$ to $\left \lfloor{\omega^2/2}\right \rfloor $. The corresponding mode profiles (for the positive frequency band) are 
\begin{equation}
    \begin{pmatrix}
        \phi_1^n(x)\\\phi_2^n(x)
    \end{pmatrix} =  \frac{1}{\sqrt{2\omega_n}}\alpha_n\begin{pmatrix}
        \sqrt{\frac{2n}{\omega_n\pm k_y^n}}f_n(x)\\-i\sqrt{\omega_n\pm k_y^n}f_{n-1}(x)
    \end{pmatrix}\,. 
\end{equation}

The corresponding Wigner matrix 
\begin{equation}\label{DiracWigner}
    W^n_{ij}  = \frac{1}{2\pi}\int \phi_i(x+s/2)\phi_j^*(x-s/2)e^{-ik_xs} ds\,,
\end{equation}
can be calculated analytically as follows. First of all, notice that it can be expressed through two functions 
\begin{equation}
    g_n(x,k_x) = \frac{1}{2\pi}\int f_n(x+s/2)f_n^*(x-s/2)e^{-ik_xs}ds\,,
\end{equation}
and 
\begin{equation}
    h_n(x,k_x) = \frac{1}{2\pi}\int f_n(x+s/2)f_{n-1}^*(x-s/2)e^{-ik_xs}ds\,.
\end{equation}
The expression for $g_n$ is well known \cite{folland2016harmonic}: 
\begin{equation}
    g_n(x,k_x) = \frac{(-1)^n}{\pi}e^{-(x^2+k_x^2)}L_n(2(x^2+k_x^2))\,.
\end{equation}
The expression for $h_n$ can be obtained using the raising operator as follows:
\begin{align}
    h_n(x,k_x) &= \frac{1}{2\pi} \int \braket{x+s/2}{n}\braket{n-1}{x-s/2}e^{-ik_xs}ds\nonumber\,,\\
             &= \frac{1}{\sqrt{2n}} \frac{1}{2\pi} \int \bra{x+s/2}\ket{n}\bra{n} \hat{a}^\dagger \ket{x-s/2}e^{-ik_xs}ds\nonumber\,,\\
             &= \frac{1}{\sqrt{2n}}g_n(x,k_x) * a^\dagger\,.
\end{align}
Here, $*$ is the Moyal star product on symbols \cite{dodin2024quasilinear} and $a^\dagger = x-ik$ is the symbol of the raising operator. Since $a^\dagger$ is a linear function of $(x,k_x)$ the $*$ product can be calculated explicitly as 
\begin{equation}
    g_n(x,k_x) * a^\dagger(x,k_x) = g_n(x,k_x)a^\dagger(x,k_x) + \frac{i}{2}\left\{ g_n(x,k_x),a^\dagger(x,k_x)\right\}\,,
\end{equation}
where $\{\cdot \}$ is the canonical Poisson bracket. It is easily seen then that
\begin{equation}
   h_n(x,k_x) =  \frac{1}{\sqrt{2n}}\left((x+\frac{1}{2}\partial_x) -i (k_x+\frac{1}{2}\partial_{k_x})\right)g_n(x,k_x)\,,
\end{equation}
which leads to the following formulas for the Wigner matrix components 

\begin{align}\label{WignerMatrixExact}
    W_{11}^n &= \frac{1}{2\omega_n}\frac{2n}{\omega_n-k_y^n}g_n(x,k_x)\,,\nonumber\\
    W_{12}^n = ({W_{12}^n})^* &= \frac{in}{2\omega_n}\frac{x-ik_x}{x^2+k^2}(g_n(x,k_x)+g_{n-1}(x,k_x))\,,\nonumber\\
    W_{22}^n &= \frac{1}{2\omega_n}(\omega_n-k_y^n)g_{n-1}(x,k_x)\,.
\end{align}

\begin{figure}
    \centering
    \includegraphics[width=0.7\linewidth]{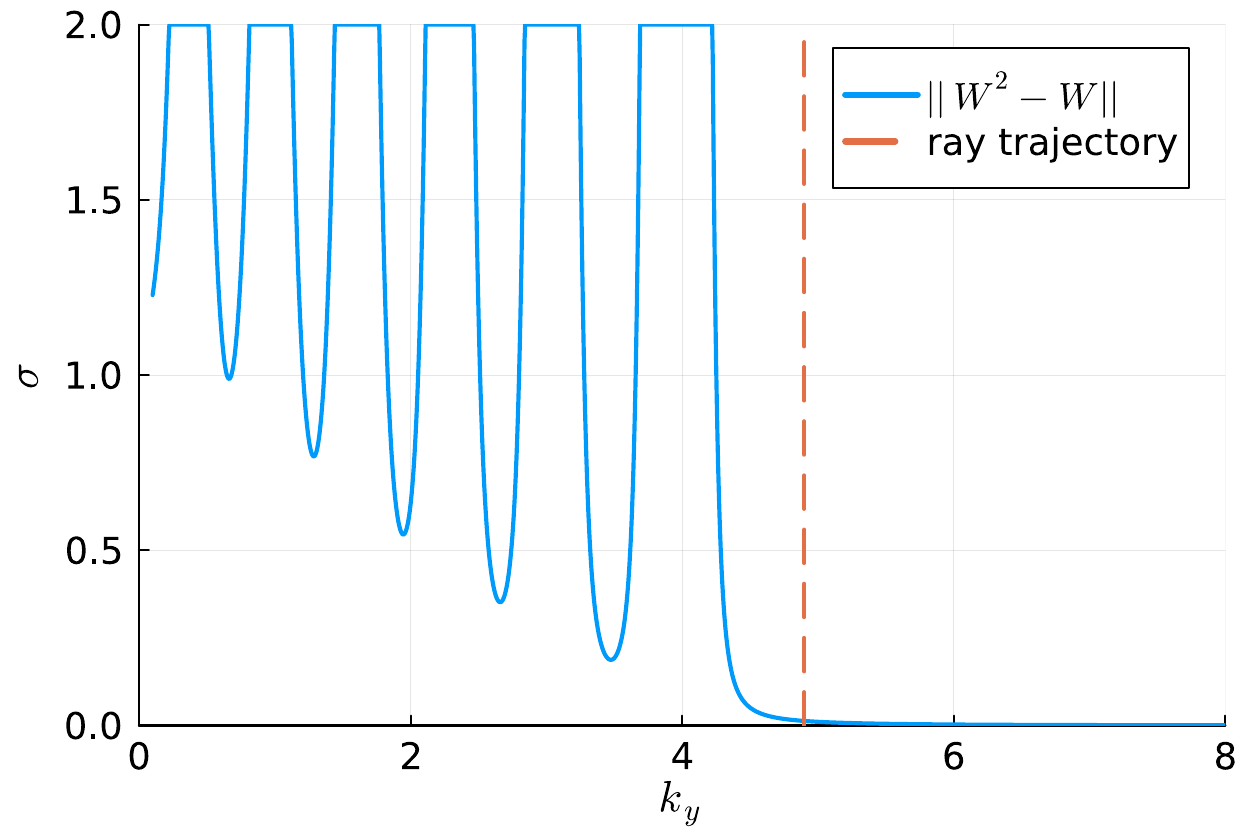}
    \caption{Plot of the error $\sigma$ \eqref{errorinW} for the $(W^{12})_{ij}$ Wigner matrix at $k_x=0$ and $\omega= 5$. The dashed orange line shows the intersection of this slice with the ray trajectory at $\omega=5$, which, for said $n=12$, corresponds to $k_y=\sqrt{24}$ by \eqref{Quantisedky}.}
    \label{fig:WignerError}
\end{figure}

\begin{rmrk}
    Given numerical numerical data for $\phi$ one can Fourier transform it in $y$ and then consider individual $k_y$ separately. This translates into individual $n$ by \eqref{Quantisedky}, so one can explore the Wigner matrices \eqref{DiracWigner} separately for separate $n$. For example, we show the error $\sigma$ determined by \eqref{errorinW}, for $n=12$ in Figure \ref{fig:WignerError}. For simplicity, we show only the slice corresponding to $k_x=0$. The GO region corresponds to the vicinity of the ray trajectory that satisfies the dispersion relation $\omega^2 = x^2+k_x^2+k_y^2$. Since the parameters for this figure are such that $\omega=5$ and $k_y=1$, zero $k_x$ corresponds to $x=\sqrt{24}$, a location marked with a red dotted line in Figure \ref{fig:WignerError}. One can see that the error $\sigma$ in this region is very small, indicating that using the eigenvector of $W$ with largest eigenvalue to estimate the local polarisation $\chi$, that is \eqref{chiestimate}, will be an accurate in this region. This is what we are expecting since that estimate followed from the GO approximation which is most accurate in the vicinity of the ray trajectory.
\end{rmrk}

We now demonstrate the general method to calculate $\text{Ch}_1$, for the $\omega>0$ band, discussed in Section \ref{Datageneral} using the Wigner matrix \eqref{WignerMatrixExact} as a model of experimental data. We will consider several different frequencies $\omega$ because it allows us to test if our $\text{Ch}_1$ results vary with the frequency. They shouldn't vary with $\omega$ but the GO approximation is more accurate at higher frequency and therefore the method of Section \ref{Datageneral} may be less reliable at low frequency. Specifically, we choose to consider the family of $\omega$'s displayed in Table \ref{tab:SyntheticWignerResults}. Note that only positive frequencies are considered since we are focusing on the $\omega>0$ band.

We need to select a sphere around the Weyl point, located at the origin in $(x,k_x,k_y)$ space, to sample our Wigner matrix over. Since, as per the remark above, the approximation $\eqref{chiestimate}$ is most accurate in the vicinity of the ray trajectory, we use the dispersion surfaces (defined by $\omega^2=x^2+k_x^2+k_y^2$) as the spheres.

Since the $k_y$ values are quantised at $\omega$, by \eqref{Quantisedky}, we cannot sample from every point on the sphere. Instead, for each $\omega$ there are a finite number\footnote{Specifically there are $2\lfloor{\omega^2/2}\rfloor$ circles to sample from.} of circular slices, indexed by $k_y^{(n)}$, of each dispersion sphere we can sample from. These circles are given by the equation $x^2+k_x^2=2|n|$ and from each of them we sample $10$ points. Note that the exact number of samples used here is arbitrary but this choice to always take the same number of samples from each circle implies that at higher $\omega$ our sphere will be sampled by more points since we have more circles on the dispersion sphere to sample from.

\begin{table}[t]
    \centering
    \begin{tabular}{c|c|c}
        $\omega$ & number of samples & $\text{Ch}_1$ \\\hline\hline
        2 & 20 & 0\\\hline
        2.5 & 40 & -1\\\hline 
        3 & 80 & -1\\\hline 
        4 & 140 & -1\\\hline 
        5 & 240 & -1\\\hline 
        6 & 340 & -1
    \end{tabular}
    \caption{Results of computing the first Chern number of the positive frequency band of \eqref{DiracEquation} from synthetic Wigner matrix measurements using the algorithm discussed in Section \ref{Section:DataAnalysis}. The expected correct value is $\text{Ch}_1=-1$, and this is reproduced by our algorithm consistently for high frequencies. }
    \label{tab:SyntheticWignerResults}
\end{table}

 From this sample of points discussed above we generated a simplicial mesh $K$, with vertices located at points we refer to as $(x^i,k_x^i,k_y^i)$ for $\simp{i}\in K$. At each point $(x^i,k_x^i,k_y^i)$, we calculated the Wigner matrix elements \eqref{WignerMatrixExact} and then computed the normalised eigenvector with largest eigenvalue, as per \eqref{chiestimate}, producing a vector we call $\chi_\simp{i}$. These vectors were assembled into an array of Stiefel matrices $\Phi = \{\Phi_\simp{i} = [\gamma(\chi_\simp{i}),\gamma(J\chi_\simp{i})]\}_{\simp{i}\in K^0}$ defining a local trivialisation of a  discrete vector bundle over $K$. Performing the same procedure on this $\Phi$ as we presented on the local trivialisations used in Section \ref{Section:NumericalExperiments} we can calculate the Chern number of this bundle. The result of computing the Chern number of $\Phi$ for each $\omega$ considered is presented in Table \ref{tab:SyntheticWignerResults}. We see that for $\omega>2$ we consistently obtained the correct result of $\text{Ch}_1=-1$ and hence confirm - in a synthetic experiment - that we can extract the Chern number of a local polarisation bundle around a Weyl point. 

Observe that the $\text{Ch}_1$ calculation failed for the $\omega=2$ test, which was the smallest value of $\omega$ tested. This was expected for two reasons. Firstly, the GO approximation upon which this entire approach is based is most accurate at high frequency and so we should expect the Wigner matrix to less well approximate the projector onto the local polarisations at smaller $\omega$. Secondly, at $\omega=2$ only the $n=\pm1$ modes are excited and so we only have a very small number of samples $20$ from which to construct the bundle. We have found empirically that our algorithm requires more that $\approx 50$ samples around a sphere to return consistently return the correct $\text{Ch}_1$ value and so should not expect $20$ to be sufficient. This is however only a heuristic result and may not always apply.


For a family of $\omega$'s displayed in Table \ref{tab:SyntheticWignerResults} we generated samples of $(x,k_x,k_y)$ spheres as described above. Specifically, from each circle $x^2+k_x^2=2n$ we took 20 samples. At each of these points in $(x,k_x,\lambda)$ we evaluated the eigenvector $\mathbf{e}_1$ of the absolute largest eigenvector of $W^n(x,k_x,k_y^n)$ and assembled a discrete eigenbundle from them. We then computed $\text{Ch}_1$ for each of these eigenbundles. The results of this computation are reported in the aforementioned table. We see that for $\omega>2$ we consistently obtained the correct result of $\text{Ch}_1=-1$ and hence confirmed that in a synthetic experiment that we can extract the index of a Weyl symbol monopole from an experimentally measured field. 

Note that the $\text{Ch}_1$ calculation failed for the $\omega=2$ test, which was the smallest value of $\omega$ tested. This was expected for two reasons. Firstly, the GO approximation upon which this entire approach is based is most accurate at high frequency and so we should expect the Wigner matrix to less well approximate the projector onto the eigenvectors of the Weyl symbol at smaller $\omega$. Secondly, at $\omega=2$ only the $n=\pm1$ modes are excited and so we only have a very small number of samples $20$ from which to construct the eigenvector. We have found experimentally that our algorithm requires more that $\approx 20$ samples around a sphere to return consistently return the correct $\text{Ch}_1$ value. This is however only a heuristic result and may not always apply.

\section{Summary}\label{Section:Conclusion}

In summary, we have presented an algorithm for counting the expected number of TIMs in a linear wave problem using simplicial characteristic classes. The algorithm in question does this by computing the first Chern number of a complex line bundle via the construction of a discrete approximate vector bundle. We demonstrated that our approach accurately computes the Chern number of the tautological complex line bundle and accurately predicts the number of boundary modes in the SW equations and as well as the TLCW. We also constructed an approach to extracting the Chern number from experimentally measured wave fields using Wigner transforms, and tested this method synthetic data derived from the 2D Dirac equation.

\section*{Acknowledgements}
We would like to thank L. Scocolla and V. Robins for helpful discussions. 
\section*{Funding}
This work was supported by the U.S. DOE through Contract No. DE-AC02-09CH11466.

\appendix
\section{Simplicial (co)homology review}\label{AppendixA:TDA}

Here we present a brief review of simplicial (co)homology and describe the notation we adopt in our use of this mathematical toolbox. Our review is intentionally brief as detailed discussions of the relevant mathematics can be found in the texts of Ghrist \cite{ghrist2018homological,ghrist2014elementary}, Edelsbrunner \cite{edelsbrunner2010computational}, and Hatcher \cite{Hatcher}. 

\subsection{Simplicial homology}

The key idea of simplicial (co)homology is to translate questions about the global topology of a simplicial mesh $K$ into questions about linear algebra over a chosen field $\F$. $\Z_p$ (the integers modulo a prime number $p$) and $\R$ are standard choices for $\F$. To do this we define a set of vector spaces $C_p$ using the $p$-simplices of $K$ as a basis and call elements of these vectors \textit{chains}. So a $p$-chain $\alpha$ is a formal sum of the form 
\begin{equation*}
    \alpha = \sum_{i_1<i_2<\ldots <i_p\in K^0}\alpha_{i_1\ldots i_p} \, (i_1i_2\ldots i_p)\,,
\end{equation*}
where the $\alpha_{i_1\ldots i_p}$'s indicate the $\F$ valued coefficients. Note that, it is standard to refer to $C_p$ as the $p$'th \textit{chain group} and to $p$ as its \textit{grading}. Since the $p$-simplices of $K$ are the defining basis, and they have a definite order given by the lexicographical ordering, we can also think of this chain as a column vector 
\begin{equation*}
    \alpha = \begin{pmatrix}
        \alpha_1\\\vdots\\\alpha_{N_p}
    \end{pmatrix}\,,
\end{equation*}
where $N_p$ refers the number of p-simplices in $K$.

\begin{rmrk}
In the above, we introduced the coefficients $\alpha_{i_1\ldots i_p}$ with the assumption that they belong to a field $\F$. This is not strictly necessary as we can relax this requirement to allow to the coefficients to be taken from any ring, such as $\Z$. For simplicity, we will focus on the field case here but will note the distinction when relevant. 
\end{rmrk}

 The topology of the mesh $K$ can be encoded algebraically into linear maps, called the boundary maps $B_p:C_p\rightarrow C_{p-1}$, which are defined to take each $p$-simplex to the formal sum of its topological boundary:
\begin{equation}\label{BoundaryMatrix}
    B_p(i_1\ldots i_p) = \sum_{j=1}^p (-1)^j(i_1\ldots \hat{i}_j \ldots i_p)\,,
\end{equation}
where the $\hat{i}_j$ indicates removing the index $i_j$ from the ordered list specifying the simplex. 

We now have a sequence of vector spaces connected by linear maps (called a chain complex) of the form
\begin{equation*}
    \cdots \overset{B_{p+2}}{\longrightarrow} C_{p+1} \overset{B_{p+1}}{\longrightarrow} C_p \overset{B_p}{\longrightarrow} C_{p-1}  \overset{B_{p-1}}{\longrightarrow}\cdots\,.
\end{equation*}
In the column vector language from before each $B_p$ is a matrix of size $N_{p-1}\times N_p$ with $\F$ valued entries, for which \eqref{BoundaryMatrix} specifies the columns. 

The core idea of homology is that topological features of spaces, specifically the \textit{number of holes of dimension $p$}, are encoded in the null vectors of $B_p$, which are called \textit{cycles}. However, there are many trivial cycles which need to be removed for this to be a useful topological invariant. These trivial cycles occur because the boundary matrices satisfy the relation $B_{p}B_{p+1}=0$ for all $p$. Therefore, if a vector $\alpha\in C_p$ is the \textit{boundary} of a vector $\beta\in C_{p+1}$ it will be a null vector $B_p\alpha = B_pB_{p+1}\beta = 0$. These trivial null vectors can be removed by introducing a sequence of quotient vector spaces, called the \textit{homology}, of the simplicial complex
\begin{equation}
    H_p(K,\F) = \frac{\ker B_p}{\Im B_{p+1}}\,.
\end{equation}
Since $H_p$ is a quotient space, each vector in $H_p$ is an equivalence class of vectors in $\ker B_p\subset C_p$, where the equivalence relation is defined as \textit{vector equality modulo boundaries}. By which we mean that if $\alpha,\beta\in C_p$ then $\alpha \sim \beta$ iff $\alpha = \beta + B_{p+1}\gamma$ for some vector $\gamma \in C_{p+1}$. 

Usually we refer to homology classes with square brackets $[\alpha]\in H_*$ where $\alpha\in C_*$ refers to a representative simplicial chain. Note that we treat this convention as a soft rule and will abuse it when there is no risk of confusion by dropping the brackets.

\subsection{The fundamental class}

Geometrically, $H_p$ counts \textit{the number of holes} in the complex in dimension $p$. In the sense that its vector space dimension encodes this count. For example, for a mesh $K$ triangulating a compact two dimensional surface, which is the main concern of this paper, we know that it encloses a three dimensional void. This void is interpreted as a two dimensional hole and as such $H_2(K,\F) = \F$ so that the dimension of the vector space (1) equals to the number of 2D holes (1). Note that, in this case, the basis vector generating this single copy of $\F$ is called \textit{the fundamental class} $[\mu_K] \in H_2(K,\F)$, and it plays an important role in defining consistent orientations on the surface of $K$.

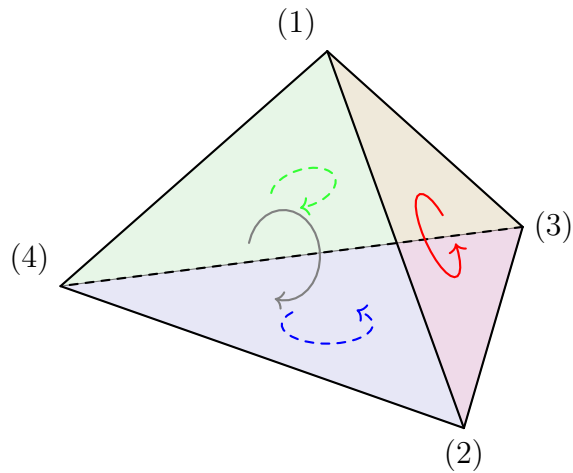
\begin{wrapfigure}{r}{0.5\textwidth}
    \begin{tikzpicture}[line cap=round, line join=round>=stealth,scale=4,3d view={130}{25}]

    \coordinate (y1) at (0,0,1);
	\coordinate (y2) at (0.333,0.866,0);
	\coordinate (y3) at (-1,0,0);
    \coordinate (y4) at (0.333,-0.866,0);

	\path[draw=black, fill = red!40, opacity = 0.3] (y1) -- (y2) -- (y3) ;

    \path[draw=black, fill = gray!30, opacity = 0.3] (y1) -- (y2) -- (y4) ;

    \path[draw=black, fill = green!20, opacity = 0.3] (y1) -- (y3) -- (y4) ;

    \path[draw=black, fill = blue!20, opacity = 0.3] (y2) -- (y3) -- (y4) ;
    
    \draw [thick, black] (y1) -- (y2);
    \draw [thick, black] (y1) -- (y3);
    \draw [thick, black] (y2) -- (y3);
    \draw [thick, black] (y1) -- (y4);
    \draw [thick, black] (y2) -- (y4);
    \draw [thick, black,dashed] (y3) -- (y4);

    \node[above left] (i) at (y1) {$\simp{ 1}$};
    \node[below] (j) at (y2) {$\simp{ 2}$};
    \node[right] (k) at (y3) {$\simp{ 3}$};
    \node[above left] (w) at (y4) {$\simp{ 4}$};

    \draw [blue,thick,dashed,domain=180:450,->] plot (
    {0+0.15*sin(-\x)}, 
    {0+0.15*cos(-\x)},
    {0});

    \draw [red,thick,domain=180:500,->] plot (
    {-0.222+ 0.15*( 0.292*sin(-\x)+0.828*cos(-\x)}, 
    { 0.289+ 0.15*(-0.379*sin(-\x)+0.559*cos(-\x)},
    { 0.333+ 0.15*( 0.877*sin(-\x)-0.033*cos(-\x)});

    \draw [green!80,thick,dashed,domain=180:450,->] plot (
    {-0.222+ 0.11*( 0.293*sin(-\x)-0.829*cos(-\x)}, 
    {-0.289+ 0.11*( 0.380*sin(-\x)+0.560*cos(-\x)},
    { 0.333+ 0.11*( 0.878*sin(-\x)+0.033*cos(-\x)});

    \draw [gray,thick,domain=180:450,->] plot (
    { 0.222+ 0.15*(-0.316*sin(-\x)+0.000*cos(-\x)}, 
    { 0.000+ 0.15*(-0.000*sin(-\x)+1.000*cos(-\x)},
    { 0.333+ 0.15*( 0.949*sin(-\x)-0.000*cos(-\x)});

    \end{tikzpicture}
    \caption{\small{Example surface mesh for a tetrahedron. The orientation of each face is indicated by an oriented loop whose color is coded to the colour of the associated face.}}
    \label{fig:tetrahedron}
\end{wrapfigure}

Given the matrix representation for $B_2$ discussed above, calculating a representative $\mu_K$ for the fundamental class is straightforward. A mesh $K$, triangulating such a surface, is comprised of only $0$,$1$, and $2$-simplices, so $C_p=0$ for $p\geq3$ and therefore the chain complex simplifies to
\begin{equation*}
       0 \overset{B_{3}}{\longrightarrow}C_2 \overset{B_2}{\longrightarrow} C_1 \overset{B_1}{\longrightarrow} C_0  \overset{B_0}{\longrightarrow}0\,,
\end{equation*}
where $0$ refers to the zero dimensional vector space. Since there are no $3$-simplices in $K$, we know that $\Im B_3 =0$ and so $H_2(K,\F) = \ker{B_2} = \F$. Therefore, $B_2$ must have a one dimensional null space the elements of which belong to the fundamental class $[\mu_K]$. So, finding  $\mu_K$ requires only finding a basis for the 1 dimensional null space of the matrix $B_2$, which can just be done with gaussian elimination\footnote{Although it is more common in practice to instead compute the \textit{smith normal form} of $B_2$ since this can be done with coefficients taken from any ring rather than just any field \cite{munkres2018elements}.}.

To see the role that $\mu_K$ plays in orienting the surface, consider the case of the surface of a tetrahedron as shown in Figure \ref{fig:tetrahedron}. This complex has six edges $\{(12),(13),(14),(23),(24),(34)\}$ and four faces $\{(123),(124),(134),(234)\}$ as the basis for its $C_1$ and $C_2$ chain groups. Note that the assumed orientations of  faces are indicated in Figure \ref{fig:tetrahedron} by the small curved arrow drawn onto them. Applying the right hand rule to these arrows implies that two of the faces are inwards pointing and two are outwards pointing. Therefore, the tetrahedron is not consistently oriented. 

The fundamental class contains the information to construct a consistent orientation. To see this, consider the $B_2$ boundary matrix written with respect to the bases defined above. As a matrix this is 
\begin{equation}
    B_2 = \begin{pmatrix}
        1 & 1 & 0 & 0\\
        -1 & 0 & 1 & 0\\
        0 & -1 & -1 & 0\\
        1 & 0 & 0 & 1\\
        0 & 1 & 0 & -1\\
        0 & 0 & 1 & 1
    \end{pmatrix}\,,
\end{equation}
and we can read off the null vector, and hence representative of the fundamental class, as 
\begin{equation}
    \mu_K = \begin{pmatrix}
        1 \\ -1 \\ 1 \\ -1
    \end{pmatrix} = (123)-(124)+(134)-(234)\,.
\end{equation}
Observe that in $\mu_K$, the two basis chains with $-1$ coefficients are $(124)$ and $(234)$ and that if we invert the orientation on these two faces in Figure \ref{fig:tetrahedron} then the surface will have a consistent outwards pointing orientation. This observation generalises to imply that the representative of the fundamental class specifies a global orientation on a simplicial mesh of a surface.  

\subsection{Simplicial cohomology}

Simplicial homology comes paired with an algebraic dual theory called \textit{cohomology}. This dual theory can be constructed by considering the sequence of dual vector spaces of the above vector spaces of chains, which we refer to as $C^p = (C_p)^*$. We call the elements of this vector space simplicial \textit{cochains} and always work in the natural dual basis of the fundamental simplicial basis defined above, where dual elements are notated with asterisks. So, for example, we use the symbol $(i)^*$ to refer to the 0-cochain dual to $(i)$,. We can then represent cochains as column vectors in this basis similarly to what we do with chains. 

We think of chain and cochains as equipped with a billinear \textit{pairing} $\langle\cdot,\cdot\rangle:C^p\times C_p\rightarrow\F$ which encodes the duality between them. Specifically, a basis $p$-chain when \textit{paired} with its own dual will return $1\in \F$ and otherwise will return zero. So, for example, we use the symbol $(i)^*$ to refer to the 0-cochain dual to $(i)$, which is defined to satisfy $\pair{(i)^*}{(j)} = \delta_{ij}$. This extends linearly over all chains and cochains in each grading $p$. If we represent chains and cochains by vectors $\alpha\in C_p$ and $\beta\in C^p$ respectively, as discussed above, then this pairing just reduces to the vector product

\begin{equation}
    \pair{\beta}{\alpha} = \beta^\intercal\alpha\,.
\end{equation}

We can equivalently think of a $p$-cochain as an assignment of value in $\F$ to each $p$-chain in $C_p$. For example, suppose we assign a value $\gamma_\simp{i}\in \F$ to each $\simp{i}\in K^0$. Then we can define a $0$-cochain $\gamma\in C^0(K,\F)$ by requiring that its pairing with each $\simp{i}$ reproduces $\gamma_{\simp{i}}$. That is, we define
\begin{equation}
    \pair{\gamma}{\simp{i}} = \gamma_{\simp{i}}\,,
\end{equation}
and then extend linearly over chains. This is a one-to-one correspondence and it how we encode cochains numerically Section \ref{Section:EulerClass}. So, for example, in the algorithm \verb!eu!, which calculates a cochain $e\in C^2(K,\Z)$, representing the Euler class of a vector bundle, $e$ is stored as an array $e = \{e_{\simp{ijk}}\}_{\simp{ijk}\in K^2}$ assigning an integer $e_\simp{ijk}$ to each $2$-simplex in the mesh.

The cochain vector spaces are then connected to each other by the transpose of the boundary maps, which for triangulated surface results in the sequence
\begin{equation*}
    0 \overset{B_{3}^\intercal}{\longleftarrow}C^2 \overset{B_2^\intercal}{\longleftarrow} C^1 \overset{B_1^\intercal}{\longleftarrow} C^0  \overset{B_0^\intercal}{\longleftarrow}0\,.
\end{equation*}
Note that we refer to $B_p^\intercal:C_{p-1}\rightarrow C_p$ as the \textit{coboundary} matrices and that, like with the boundary matrix, it composes with itself to give zero $B_{p+1}^\intercal B_{p}^\intercal = (B_{p}B_{p+1})^\intercal = 0$. Importantly, the boundary and coboundary matrices are adjoint with respect to the pairing $\pair{\cdot}{\cdot}$. That is, if $\alpha\in C_{p+1}$ and $\beta\in C^{p}$ then
\begin{equation}
    \pair{\beta}{B_{p+1}\alpha} = \beta^\intercal B_{p+1}\alpha = (B_{p+1}^\intercal \beta)^\intercal \alpha = \pair{B_{p+1}^\intercal \beta}{\alpha}\,.
\end{equation}

Similarly to homology, in cohomology topological features can be encoded in cochains which vanish under $B_p^\intercal$, called \textit{cocycles}, but which are not the coboundary of another cochain. We therefore define quotient vector spaces 
\begin{equation}
    H^p(K,\F) = \frac{\ker B_{p+1}^\intercal}{\Im B_{p}^\intercal}\,,
\end{equation}
called the simplicial \textit{cohomology} of $K$. As before we refer to cohomology classes with square brackets $[c]\in H^*$ where $c\in C^*$ refers to a representative simplicial cochain. Again, we treat this as a soft convention. 

\begin{rmrk}
    It turns out that the simplicial homology and cohomology are isomorphic in each grading $H_p = H^p$ but cohomology carries extra algebraic structure. Specifically, cohomology classes can be multiplied using an operation called the cup product, notated as $[\alpha]\smile[\beta]$ for $[\alpha],[\beta]\in H^*$. The details of how this product is constructed are not relevant for us here but can be found in standard sources on algebraic topology such as the text by Hatcher \cite{Hatcher}. We mention the cup product because it appears in the definition of the Stiefel--Whitney and Euler classes presented in Section \ref{subsec:SWAndEuler}.
\end{rmrk} 

Another note worthy of mention is that the pairing $\pair{\cdot}{\cdot}$ discussed above lifts to a pairing between homology and cohomology classes. If $[\alpha]\in H_p$ and $[\beta]\in H^p$ then this pairing is defined by
\begin{equation}
    \pair{[\beta]}{[\alpha]} = \pair{\beta}{\alpha} = \beta^T\alpha\,.
\end{equation}
The above is well defined because it does not depend upon the representative chosen. To see this suppose that we had chosen to different representatives $\alpha' = \alpha + B_{p+1}\gamma$ and $\beta' = \beta + B_p^\intercal\delta$, then it follows that
\begin{equation*}
    \pair{[\beta']}{[\alpha']} = (\beta+B_{p}^\intercal\delta)^\intercal(\alpha+B_{p+1}\gamma) =\beta^\intercal\alpha + \delta^\intercal B_p\alpha + (B_{p+1}^\intercal\beta)^\intercal\gamma + \delta^\intercal B_{p}B_{p+1}\gamma\,,
\end{equation*}
but only the first term in the above remains since $\alpha$ is a cycle and $\beta$ is a cocycle and the boundary matrices compose to zero. Therefore 
\begin{equation*}
    \pair{[\beta']}{[\alpha']} =\beta^\intercal\alpha  = \pair{[\beta]}{[\alpha]}\,,
\end{equation*}
and hence this pairing is independent of the representative.

\printbibliography

\end{document}